\documentclass[english,11pt,a4paper,dvipsnames]{article}
\ifx\pdfoutput\undefined\else
\pdfoutput=1
\fi 
\usepackage[T1]{fontenc}
\usepackage[latin9,utf8]{inputenc}
\usepackage{geometry}
\usepackage{babel}
\usepackage{hyperref}
\usepackage[shortlabels]{enumitem}
\usepackage{float}

\usepackage{amsmath, amssymb, amstext, mathtools}
\usepackage{amsthm}

\usepackage{algorithm}
\usepackage{algpseudocodex}
\usepackage{tikz}
\usepackage{pgfplots}
\usepackage{textpos}

\usepackage[capitalise]{cleveref}

\usepackage{xcolor}

\usetikzlibrary{fadings,backgrounds,patterns,arrows,decorations.pathreplacing,decorations.pathmorphing,calc}

\newtheorem{theorem}{Theorem}[section]
\newtheorem{lemma}[theorem]{Lemma}
\newtheorem{corollary}[theorem]{Corollary}

\newtheorem{observation}[theorem]{Observation}
\newtheorem{conjecture}[theorem]{Conjecture}
\newtheorem{claim}[theorem]{Claim}

\theoremstyle{definition}
\newtheorem{definition}[theorem]{Definition}
\newtheorem{example}[theorem]{Example}
\theoremstyle{remark}
\newtheorem{remark}[theorem]{Remark}

\newenvironment{claimproof}{\begin{proof}}{\end{proof}}

\newcommand{\NN}{{\mathbb N}}

\newcommand{\paths}[3]{\mathcal{P}_{#1}(#2#3)}

\newcommand{\numindsub}[2]{\#\mathrm{IndSub}(#1 \to #2)}

\newcommand{\numsub}[2]{\#\mathrm{Sub}(#1 \to #2)}
\newcommand{\numsubstar}[1]{\#\mathrm{Sub}(#1 \to \,\star\,)}

\newcommand{\pcliquedecision}{\textnormal{\textsc{Clique}}}

\newcommand{\sub}[1]{\textnormal{\textsc{Sub}}(#1)}
\newcommand{\colsub}[1]{\textnormal{\textsc{ColSub}}(#1)}

\newcommand{\csub}[1]{\#\textnormal{\textsc{Sub}}(#1)}
\newcommand{\ccolsub}[1]{\#\textnormal{\textsc{ColSub}}(#1)}

\newcommand{\sindsub}[1]{\#\textnormal{\textsc{IndSub}}(#1)}

\newcommand{\threecol}{\textnormal{\textsc{3-Coloring}}}
\newcommand{\threeass}{\textnormal{\textsc{3-Assignment}}}
\newcommand{\countthreeass}{\#\textnormal{\textsc{3-Assignment}}}
\newcommand{\countthreecol}{\#\textnormal{\textsc{3-Coloring}}}

\newcommand{\can}[1]{#1^{\textnormal{\textrm{id}}}}

\DeclareMathOperator{\tw}{tw}

\DeclareMathOperator{\supp}{supp}

\newcommand{\uparrowsm}{{\scriptscriptstyle\uparrow}}
\newcommand{\downarrowsm}{{\scriptscriptstyle\downarrow}}

\newcommand{\circleuv}{%
    \begin{tikzpicture}[baseline=(uv.base)]
        \draw[fill=gray!30, draw=none] (0,0.07) circle [radius=0.2cm];
        \node at (0,0.04) (uv) {\scriptsize $uv$};
    \end{tikzpicture}%
}

\newcommand*\circlednum[1]{\tikz[baseline=(char.base)]{
            \node[shape=circle,draw,inner sep=1.2pt] (char) {#1};}}

\definecolor{cbfp1}{RGB}{120,94,240}
\definecolor{cbfp2}{RGB}{220,38,127}
\definecolor{cbfp3}{RGB}{254,97,0}
\definecolor{cbfp4}{RGB}{255,176,0}

\newcommand{\orcid}[1]{\href{https://orcid.org/#1}{\includegraphics[height=1.8ex]{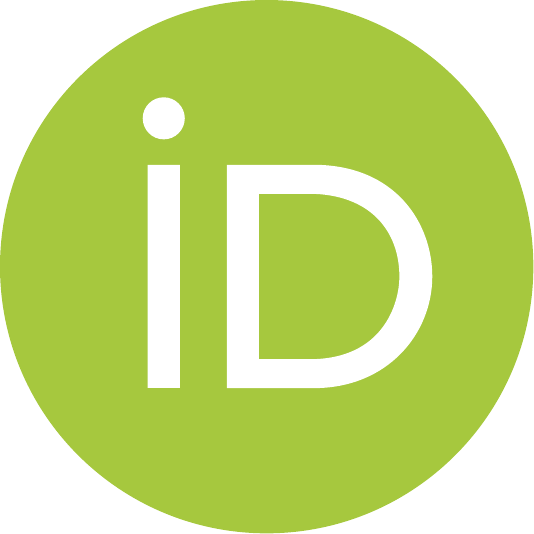}}}

\usepackage{microtype}

\makeatletter
\def\blfootnote{\gdef\@thefnmark{}\@footnotetext}
\makeatother

\title{Can You Link Up With Treewidth?%
\blfootnote{\rightskip=5.7cm
The research is funded by the European Union (ERC, CountHom, 101077083).
Views and opinions expressed are those of the author(s) only and do not necessarily reflect those of the European Union or the European Research Council Executive Agency. Neither the European Union nor the granting authority can be held responsible for them.
}
\blfootnote{\rightskip=5.7cm
The title was found with the help of a popular LLM.
We thank Cornelius Brand for pointing out a connection to extension complexity.
}
}
\author{
Radu Curticapean \orcid{0000-0001-7201-9905} \\
University of Regensburg and IT University of Copenhagen
\and
Simon D\"{o}ring \orcid{0009-0002-6667-5257} \\
Max Planck Institute for Informatics and Saarland University
\and
Daniel Neuen \orcid{0000-0002-4940-0318}\\
University of Regensburg and Max Planck Institute for Informatics
\and
Jiaheng Wang \orcid{0000-0002-5191-545X}\\
University of Regensburg
}
\date{}

\definecolor[named]{urlblue}{cmyk}{1,0.58,0,0.21}
\hypersetup{
    breaklinks=true,
    colorlinks=true,
    citecolor=purple!70!blue!60!black,
    linkcolor=purple!70!blue!90!black,
    urlcolor=urlblue,
    pdflang={en},
    pdftitle={Can You Link Up With Treewidth?},
    pdfauthor={Radu Curticapean, Simon Döring, Daniel Neuen, Jiaheng Wang}
}

\begin{document}

\maketitle
\begin{textblock}{5}(7.85, 7.55) \includegraphics[width=150px]{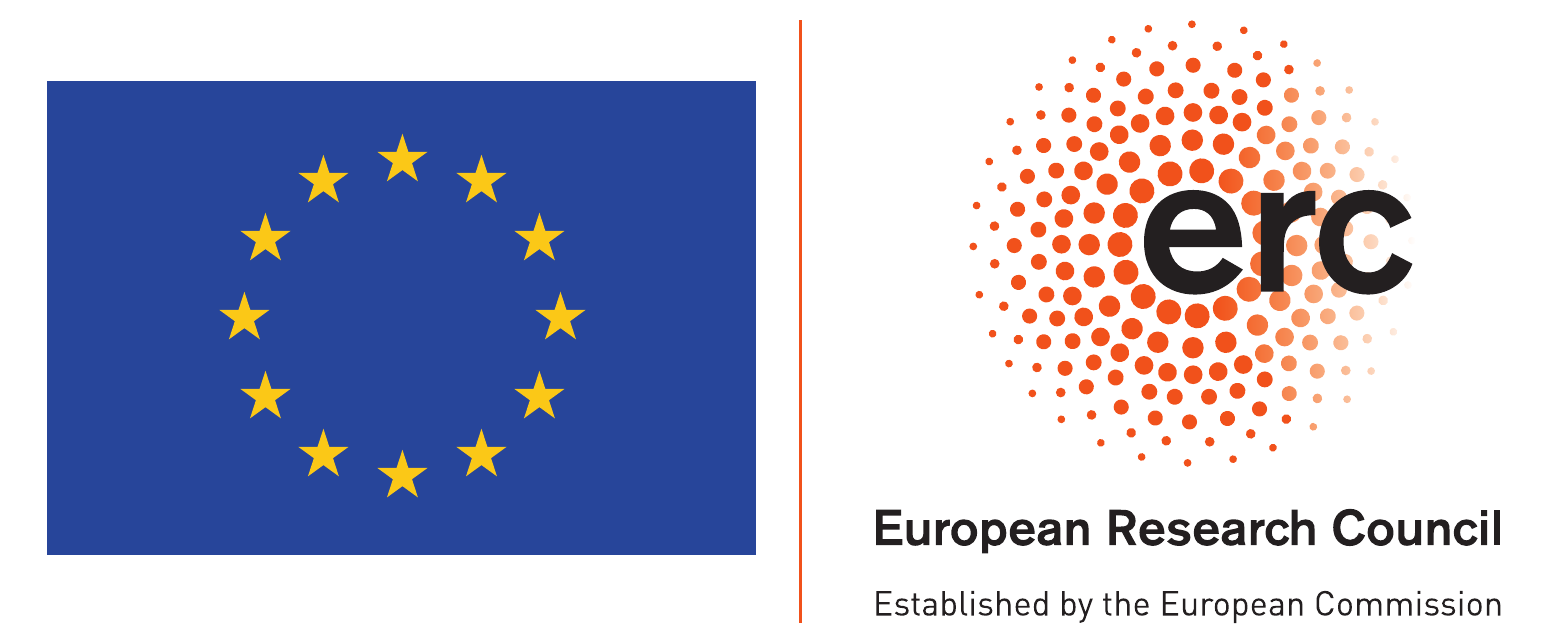} \end{textblock}

\begin{abstract}
    In a fundamental paper in parameterized complexity theory, Marx [ToC~'10] constructed $k$-vertex graphs $H$ of maximum degree $3$ such that $n^{o(k /\log k)}$ time algorithms for detecting colorful $H$-subgraphs would refute the Exponential-Time Hypothesis (ETH).
    This result is widely used to obtain almost-tight conditional lower bounds for parameterized problems under ETH.

    We give a new and fully self-contained proof of this result that further simplifies a recent work by Karthik et al.\ [SOSA 2024].
    In our proof, we introduce a novel graph parameter of independent interest, the \emph{linkage capacity} $\gamma(H)$, and show that detecting colorful $H$-subgraphs in time $n^{o(\gamma(H))}$ refutes ETH.
    Then, we use a simple construction of communication networks credited to Bene\v{s} to obtain $k$-vertex graphs of maximum degree $3$ and linkage capacity $\Omega(k / \log k)$, avoiding arguments involving expander graphs, which were required in previous papers.
    We also show that every graph $H$ of treewidth $t$ has linkage capacity $\Omega(t / \log t)$, thus recovering a stronger result shown by Marx~[ToC~'10] with a simplified proof.
    
    Additionally, we obtain new tight lower bounds on the complexity of colorful subgraph detection for certain types of patterns by analyzing their linkage capacity:
    We prove that almost all $k$-vertex graphs of polynomial average degree $\Omega(k^{\beta})$ for $\beta > 0$ have linkage capacity $\Theta(k)$, which implies tight lower bounds for finding such patterns $H$.
    As an application of these results, we also obtain tight lower bounds for counting small induced subgraphs having a fixed property $\Phi$, improving bounds from, e.g., [Roth et al., FOCS 2020].
\end{abstract}

\clearpage

\section{Introduction}

Over the past two decades, it has been shown that complexity assumptions about \emph{exponential-time} problems imply far-reaching lower bounds for \emph{polynomial-time} \cite{Bringmann19,Williams15,Williams18} and \emph{parameterized} \cite{CyganFKLMPPS15,LMS11} problems.
Among the first such results, it was shown that the Exponential-Time Hypothesis (ETH) about the Boolean satisfiability problem $\textsc{SAT}$ implies an $n^{\Omega(k)}$-time lower bound for the seemingly unrelated parameterized problem $\pcliquedecision$ of detecting $k$-cliques in $n$-vertex graphs \cite{ChenCFHJKX05,ChenHKX06}.
This lower bound solidifies the status of $\pcliquedecision$ as a canonical hard problem in parameterized complexity.

Ideally, when reducing $\pcliquedecision$ to some target problem, we would like to transfer the known $n^{\Omega(k)}$-time lower bound for $\pcliquedecision$ under ETH to the target problem.
However, reductions from $\pcliquedecision$ often require $k$ gadgets to encode the vertices of a $k$-clique and $\Theta(k^2)$ additional gadgets to verify the edges between all pairs of encoded vertices.
As each gadget typically increases the parameter by at least a constant amount, instances for $\pcliquedecision$ are transformed into target instances with a parameter value of $\Theta(k^2)$ (see, e.g., \cite[Section 13.6.3]{CyganFKLMPPS15}).
This in turn means that only $n^{o(\sqrt \ell)}$-time algorithms can be ruled out for a target problem with parameter $\ell$.

Tighter lower bounds could be obtained if we could reduce from a subgraph problem similar to $\pcliquedecision$, but involving $k$-vertex patterns $H$ with only $O(k)$ rather than $\Theta(k^2)$ edges.
More specifically, for a fixed graph $H$, let $\colsub{H}$ be the problem of detecting $H$-subgraph copies in graphs $G$ with vertex-colors from $V(H)$ such that every $v \in V(H)$ is mapped to a vertex of color $v$ in $G$. (This problem can equivalently be interpreted as a \emph{constraint satisfaction problem} with variables $x_v$ for $v \in V(H)$ and arity-$2$ relations $R_e$ for $e\in E(H)$. The domain of $x_v$ is the set of $v$-colored vertices in $G$.)
Many known parameterized reductions from $\pcliquedecision$ can be modified to use $\colsub{H}$ as the reduction source, and a seminal result by Marx~\cite[Corollary 6.1]{Marx10} 
shows that $\colsub{H}$ is indeed hard under ETH for graphs $H$ of maximum degree $3$, albeit not with an entirely tight lower bound:

\begin{theorem}[\cite{SMPS24, Marx10}]
    \label{thm:CYBT-sparse}
    Assuming ETH, there exists a universal constant $\alpha > 0$ and an infinite sequence of graphs $H_1,H_2,\ldots$ such that, for all $k \in \NN$, the graph $H_k$ has $k$ vertices and maximum degree $3$,
    and $\colsub{H_k}$ does not admit an $O(n^{\alpha \cdot k / \log k})$-time algorithm.
\end{theorem}

This theorem has become a standard tool to prove almost-tight lower bounds along the lines of the above reduction scheme, and it has been applied to numerous parameterized problems from a diverse range of areas~\cite{APSZ21,AivasiliotisGR24,AKMR19,BCMP20,DBLP:journals/jocg/BonnetGL19,BIJK19,DBLP:journals/talg/BonnetM20,DBLP:journals/dam/BonnetS17,DBLP:conf/esa/Bringmann0MN16,CFM21,XFHM20,CVMM21,CCG13,CDH21,CurticapeanDM17,DBLP:conf/focs/CurticapeanX15,DEG22,EGN23,DBLP:conf/stacs/EibenKPS19,DBLP:conf/iwpec/EppsteinL18,FPRS23,DBLP:journals/algorithmica/GuoHNS13,DBLP:journals/jcss/JansenKMS13,DBLP:journals/siamdm/JonesLRSS17,KSS22,DBLP:conf/soda/LokshtanovR0Z20,NS22,DBLP:journals/toct/PilipczukW18a,RothSW24}.\footnote{\cref{thm:CYBT-sparse} follows from a more general result proved by Marx~\cite{Marx10}:
Assuming ETH, there exists a constant $\alpha > 0$ such that, for every fixed graph $H$ with treewidth $t$, the problem $\colsub{H}$ cannot be solved in time $O(n^{\alpha \cdot t / \log t})$.
Invoking this theorem with $3$-regular expander graphs yields \cref{thm:CYBT-sparse}.
We defer the discussion of this more general result to \cref{sec:intro-linkage-capacity-vs-treewidth}.
For most of the applications cited above, the above corollary suffices.}

The lower bound in this theorem can also be transferred to the \emph{uncolored} subgraph isomorphism problem; see \cite[Corollary 6.4]{Marx10}.
For a fixed graph $H$, let $\sub{H}$ be the problem of detecting $H$-subgraph copies in uncolored graphs $G$.

\begin{corollary}
    \label{cor:uncolored-sparse}
    Assuming ETH, there exists a universal constant $\alpha > 0$ and an infinite sequence of graphs $H_1,H_2,\ldots$ such that, for all $k \in \NN$, the graph $H_k$ has $\Theta(k)$ vertices and maximum degree $5$,
    and $\sub{H_k}$ does not admit an $O(n^{\alpha \cdot k / \log k})$-time algorithm.
\end{corollary}

\subsection{Main Concept: Linkage Capacity}

In this paper, we provide a new perspective on the seminal \Cref{thm:CYBT-sparse}, which allows us to prove new results and to obtain a significantly simpler proof, even compared to recent simplified versions~\cite{DBLP:conf/soda/JaffkeLMPS23,SMPS24}.
Our new perspective hinges upon a new graph parameter, the \emph{linkage capacity} $\gamma(H)$ of a graph $H$.
Roughly speaking, this parameter measures how well vertices of $H$ can be connected by vertex-disjoint paths on specified endpoint pairs.

\paragraph*{Known Lower Bound for Cliques.}

To explain our approach, let us first sketch the classical $n^{\Omega(k)}$-time lower bound for $\pcliquedecision$ under ETH (see, e.g., \cite[Theorem 14.21]{CyganFKLMPPS15}):
It is known that, assuming ETH, the $\threecol$ problem in $n$-vertex graphs $G$ with maximum degree $4$ cannot be solved in $2^{o(n)}$ time. We transform $G$ into an equivalent instance $X$ of $\pcliquedecision$ with approximately $3^{n/k}$ vertices; then an $n^{o(k)}$-time algorithm for $\pcliquedecision$ would imply a $2^{o(n)}$-time algorithm for the $\threecol$ problem, contradicting ETH.

To transform $G$ into $X$, the vertex set $V(G)$ is divided equitably into blocks $V_1, \ldots, V_k$. The vertices of $X$ correspond to the $3$-colorings of these blocks, and two vertices in $X$ are connected by an edge if their colorings are compatible, meaning they come from different blocks and together form a proper coloring.
This way, the $k$-cliques $K$ in this ``compatiblity graph'' $X$ correspond bijectively to valid $3$-colorings of $G$:
Indeed, the vertices of $K$ provide a valid coloring for each block, and the presence of edges between all $u, v \in V(K)$ in $X$ ensures that the union of these partial colorings is a valid coloring of the entire graph $G$.

\paragraph*{From Cliques to General Subgraphs.}

This argument can be adapted to obtain a lower bound for $\colsub{H}$ with general $k$-vertex patterns $H$. 
First, consider the favorable hypothetical scenario that the vertices of an input graph $G$ for $\threecol$ can be split equitably into blocks $V_1,\ldots,V_k$ corresponding to the $k$ vertices of $H$, such that every edge of $G$ is contained within one block or between blocks $V_i$ and $V_j$ with $ij\in E(H)$.
In other words, $G$ is contained in the \emph{blowup} $H \boxtimes K_t$ for $t \approx n/k$, which is the graph obtained from $H$ by copying every vertex to a block of $t$ clones that form a clique; see also \Cref{fig:grid-congestion}.

In this scenario, not all pairs of partial $3$-colorings need to be checked for compatibility: Instead, it suffices to check for compatibility only between vertices in blocks $V_i$ and $V_j$ satisfying $ij\in E(H)$, since no other edges are present in $G$ and therefore can not cause incompatibilities between partial colorings.

In general however, we cannot assume that $G$ fits into a moderately small blowup of $H$. 
An immediate remedy is to replace the edges in $G$ by paths on fresh vertices and edges to obtain a graph $G'$ that \emph{does} fit into a blowup of $H$.
However, this process may replace most edges by paths that visit many blocks, thus increasing the block size from $n/k$ to $n$.
Even if routing via short paths is possible, it may be possible that a few blocks are hit disproportionally often, leading to the same problem.
Both issues would render a (too) fast algorithm for $\colsub{H}$ useless for the purpose of obtaining a (too) fast algorithm for $\threecol$.

\paragraph*{Linkage Capacity.}

Towards overcoming the above problems, we observe that many $k$-vertex patterns $H$ enable an ``amortized batch-routing'' of edges in blowups $H \boxtimes K_t$.
That is, in the remedy sketched in the previous paragraph, \emph{large sets} of edges can be replaced by paths while adding only about $1/k$ vertices \emph{per edge} to each individual block.
The \emph{linkage capacity} $\gamma(H)$ quantifies how well $H$ supports such amortized routing in blowups $H \boxtimes K_t$.

More precisely, we call a vertex-set $X$ in a graph $F$ \emph{matching-linked} if, for every matching $M$ with vertices from $X$ (but with $M$ possibly containing edges not present in $F$), there exist disjoint $u$-$v$-paths in $F$ realizing the edges $uv \in M$.
The linkage capacity $\gamma(H)$ of a graph $H$ is essentially the largest $c > 0$ such that $H \boxtimes K_t$ contains a matching-linked set $X$ of size $\lfloor ct \rfloor$ for all (large enough) $t$.
We have $\gamma(H) \leq k$, since $X$ lives in $H \boxtimes K_t$, which has $kt$ vertices.

Large linkage capacity means that large matchings can be realized by vertex-disjoint paths in blowups $H \boxtimes K_t$.
This implies that graphs of bounded degree can be realized as well, as also observed in~\cite[Theorem 4.2]{SMPS24}:
For example, the $\threecol$ instances $G$ we consider have maximum degree $4$, so their edge-sets can be partitioned into $5$ matchings. 
The paths obtained for each of the five individual matchings in $H \boxtimes K_{t}$ can then be glued to find $G$ as a topological minor in $H \boxtimes K_{5t}$.
Following the above reduction sketch from $\threecol$ and using this property of graphs with large linkage capacity, we establish a conditional lower bound on the complexity of $\colsub{H}$ based on $\gamma(H)$.

\begin{theorem}
    \label{thm:linkage-ETH}
    Assuming ETH, there exists a universal constant $\alpha > 0$ such that
    no fixed graph $H$ 
    admits an $O(n^{\alpha \cdot \gamma(H)})$-time algorithm for $\colsub{H}$.
    The same holds for the counting version $\ccolsub{H}$ under the counting exponential-time hypothesis \#ETH
\end{theorem}

It remains to determine when $H$ has large linkage capacity. 
For example, if $H$ itself admits a large matching-linked set, then this translates to its blowups, thus establishing large $\gamma(H)$. 
This is however only a sufficient criterion, even though most of our lower bounds are based on it.
As we investigate in \Cref{sec:treewidth}, the linkage capacity is related to fractional multicommodity flow problems whose relevance for lower bounds on $\colsub{H}$ was already identified before \cite{SMPS24, Marx10}.
Linkage capacity however is a more elementary and more applicable concept: The restriction to matchings allows us to connect it to known results on routing with specified terminal pairs;
this in turn allows us to prove new results under ETH without much technical effort.

\subsection{Applications of Linkage Capacity}
\label{sec:intro-applications}

With \Cref{thm:linkage-ETH} in hand, we show lower bounds on the complexity of the colorful $H$-subgraph problem via linkage capacity $\gamma(H)$.
For this, we enlist the help of communication network theory~\cite{2021Benes, Benes64a}, random graph theory~\cite{BFSU96}, linear programming \cite{DBLP:journals/siamcomp/FeigeHL08,DBLP:journals/jacm/LeightonR99}, and classical results from graph theory on connectivity via vertex-disjoint paths~\cite{Mader72,TW05}.

\paragraph*{A Fully Self-Contained Proof of \Cref{thm:CYBT-sparse}.}

Our first application of \Cref{thm:linkage-ETH} is a significantly simplified and self-contained\footnote{A note on the extent to which the proof is self-contained: Our proof starts from the known result that, under ETH, the $\threecol$ problem requires $2^{\Omega(n)}$ time on $4$-regular graphs with $n$ vertices. This can be easily shown from ETH together with the Sparsification Lemma.} proof of the seminal \Cref{thm:CYBT-sparse}.
The original proof of this theorem by Marx~\cite{Marx10} uses highly nontrivial arguments regarding multicommodity flows as a black box~\cite{DBLP:journals/siamcomp/FeigeHL08}.
Even a very recent simplification~\cite{SMPS24} still requires the construction of expander graphs and routing algorithms for such graphs, both of which are highly nontrivial~\cite{DBLP:journals/combinatorica/Alon21,DBLP:journals/jacm/LeightonR99}.

By approaching the problem through linkage capacity, we observe that expansion is \emph{not} required to obtain \Cref{thm:CYBT-sparse}. Instead, we can rely on a very simple construction of telecommunication networks, credited to a 1964 paper by Bene\v{s}~\cite{Benes64a}, then employed at Bell Labs:
A \emph{Bene\v{s} network} contains $s=2^\ell$ input and output vertices, and $k = O(s \log s)$ vertices in total. 
For every pairing of inputs to outputs, the network guarantees private data streams (i.e., vertex-disjoint paths) connecting each input to its specified output.
Both the network construction and routing therein are elementary divide-and-conquer arguments that feature in undergraduate introduction courses to discrete mathematics~\cite{2021Benes}.
From such networks, we easily obtain $k$-vertex graphs of maximum degree $4$ and linkage capacity $\Omega(k / \log k)$.
Combined with \Cref{thm:linkage-ETH}, this gives a novel proof of \Cref{thm:CYBT-sparse}.

We recently found that graphs with large matching-linked sets have been used in communication and extension complexity before:
A paper by Göös, Jain, and Watson~\cite[Section~3.3]{DBLP:journals/siamcomp/Goos0018} mentions ``bounded-degree butterfly graphs'' from an unpublished manuscript on pebble games by Nordström~\cite[Proposition~5.2]{NordstromPebble} as an alternative to expanders;
this alternative construction turns out to be precisely that of Bene\v{s}.

\paragraph*{Tight Lower Bounds for Dense Graphs.}
Alon and Marx~\cite[Theorem 1.4]{DBLP:journals/siamdm/AlonM11} argue that the logarithmic slack for sparse pattern graphs in \Cref{thm:CYBT-sparse} cannot be overcome by current approaches, including ours.
More modestly, one can ask for ``just slightly'' dense $k$-vertex patterns $H$ such that $\colsub{H}$ requires $n^{\Omega(k)}$ time under ETH.

Indeed, Alon and Marx~\cite[Theorem 1.5(2)]{DBLP:journals/siamdm/AlonM11} showed that, for every $\delta > 0$, certain specifically constructed patterns $S$ with average degree $O(k^\delta)$ are particularly good for hosting subgraphs and entail $n^{\Omega(k)}$-time lower bounds on the colorful $S$-subgraph problem~\cite[Theorem 1.8]{DBLP:journals/siamdm/AlonM11}.
For some problems of interest however, e.g., for counting induced $k$-vertex patterns~\cite{CurticapeanN25,DoringMW24,RothSW24}, one can only reduce from the colorful $H$-subgraph problem for \emph{some} (say, adversarially chosen) dense pattern $H$, which may not necessarily be a graph $S$ constructed by Alon and Marx.
This imposes a bottleneck towards tight lower bounds for such problems.

One partial remedy lies in using large clique minors (see, e.g., \cite{DBLP:conf/focs/DalirrooyfardW22,RothSW24}).
Kostochka~\cite{DBLP:journals/combinatorica/Kostochka84} showed that \emph{every} graph $H$ of average degree $d$ contains a $K_q$-minor with $q = \Omega(d / \sqrt{\log d})$.
Given a $K_q$-minor in $H$, a straightforward reduction yields an $n^{\Omega(q)}$-time lower bound on the colorful $H$-subgraph problem under ETH.
This implies that \emph{every} pattern $H$ of linear average degree $\Omega(k)$ requires an exponent of $\Omega(k / \sqrt{\log k})$ for the colorful $H$-subgraph problem (see, e.g., \cite[Corollary 2.1]{DBLP:conf/focs/DalirrooyfardW22}).
While this improves upon the lower bound from \Cref{thm:CYBT-sparse}, a slack of $\Omega(\sqrt{\log k})$ remains.

Using linkage capacity, we eliminate this slack and obtain a tight lower bound for dense patterns:
Combining two textbook results~\cite{graph}, we show that every pattern $H$ of average degree $d$ has linkage capacity $\Omega(d)$.\footnote{This lower bound is asymptotically tight, since worst-case examples like $K_{d,s-d}$ have linkage capacity at most $2d+1$. Indeed, a linked set of $2d+2$ vertices would imply a linkage with $d+1$ paths in $K_{d,s-d}$. This would in particular imply a matching with $d+1$ edges, which clearly does not exist in $K_{d,s-d}$.} 
\Cref{thm:linkage-ETH} then immediately yields:

\begin{theorem}
    \label{thm:dense-hard}
    Assuming ETH, there exists a universal constant $\alpha > 0$ such that
    no fixed graph $H$ with average degree $d$
    admits an $O(n^{\alpha \cdot d})$-time algorithm for $\colsub{H}$.
    The same holds for $\ccolsub{H}$ under \#ETH.
\end{theorem}

This theorem covers the ``worst case'', i.e., patterns $H$ of fixed average degree $d$ that are adversarially chosen so as to minimize $\gamma(H)$.
In particular, for linear average degree, an $n^{\Omega(k)}$ bound under ETH follows.
This implies new tight lower bounds for very general classes of induced pattern counting problems \cite{CurticapeanN25,RothSW24} (see \Cref{sec:indsub-main} for details).

In the ``average case'', much lower density turns out to be sufficient for an $n^{\Omega(k)}$ bound.
Indeed, known results on routing in random graphs~\cite{BFSU96} imply directly that \emph{almost all} $k$-vertex graphs $H$ with average degree $d \in \Omega(k^\beta)$ for constant $\beta > 0$ have linkage capacity $\Theta(k)$.
Observe that the average degree is that of the specifically constructed patterns $S$ by Alon and Marx \cite{DBLP:journals/siamdm/AlonM11}; we show that not only specific patterns, but \emph{almost all} patterns of polynomial average degree have an $n^{\Omega(k)}$ bound for $\colsub{H}$.

More generally, we show that the linkage capacity of the Erd\H{o}s-R\'{e}nyi random graph $\mathcal G(k,p)$ for non-degenerate probabilities $p$ is $\Omega(k/\rho)$, where $\rho = \log(k) / \log(kp)$ is the typical distance between vertices in $\mathcal G(k,p)$ \cite{Bol81,KL81}.
We obtain the following general lower bound:

\begin{theorem}
    \label{thm:avg-dense-hard}
    Assuming ETH, there exists a universal constant $\alpha > 0$ such that for every constant $\varepsilon>0$ and every $p\geq (1+\varepsilon)\log k/k$, the following holds:
    With high probability, for an Erd\H{o}s-R\'{e}nyi random graph $H\sim\mathcal G(k,p)$, the problem $\colsub{H}$ does not admit an $O(n^{\alpha \cdot k / \rho})$-time algorithm.
    Here, $\rho = \log(k) / \log(kp)$ is the typical distance in $\mathcal G(k,p)$.
    The same holds for $\ccolsub{H}$ under \#ETH.
\end{theorem}

Note that $\rho$ is the logarithm of $k$ in the base of the average degree $kp$;
this captures the time needed to concurrently explore all $k$ vertices in a process that branches into $kp$ random vertices from each vertex.
It is intuitively clear that the linkage capacity should be at most $O(k/\rho)$: 
Almost all vertex pairs $u,v$ in a random graph require $u$-$v$-paths of length $\rho$, so we cannot connect more than $k/ \rho$ vertex pairs without exhausting $k$ vertices.
The bound from~\cite{BFSU96} shows that, with high probability, $\Omega(k/ \rho)$ vertex pairs \emph{can} be connected.

The last theorem also implies lower bounds for the \emph{uncolored} subgraph isomorphism problem.
A graph $H$ is a \emph{homomorphic core} if every homomorphism from $H$ to itself is injective, i.e., it is an automorphism.
It is well-known that, for homomorphic cores $H$, the colored subgraph problem $\colsub{H}$ can be reduced in linear time to its uncolored version $\sub{H}$; we include a proof in \Cref{app:uncolored}.
It is also known that a random graph $H\sim\mathcal G(k,p)$ with $k^{-1/3} (\log k)^2 <  p < 1 - k^{-1/3} (\log k)^2$ is a homomorphic core \cite{BonatoP09} with high probability.
Hence, in this parameter regime, the lower bound from \Cref{thm:avg-dense-hard} extends to the uncolored problem. 

\begin{corollary}
    \label{cor:uncolored-avg-dense-hard}
    Assuming ETH, there exists a universal constant $\alpha > 0$ such that for every $k^{-1/3} (\log k)^2 <  p < 1 - k^{-1/3} (\log k)^2$, the following holds:
    With high probability, for an Erd\H{o}s-R\'{e}nyi random graph $H\sim\mathcal G(k,p)$, the problem $\sub{H}$ does not admit an $O(n^{\alpha \cdot k})$-time algorithm.
    The same holds for $\csub{H}$ under \#ETH.
\end{corollary}

It remains an interesting open problem whether \Cref{cor:uncolored-avg-dense-hard} can be extended to the same parameter regime as \Cref{thm:avg-dense-hard}. 
Indeed, in \cite{BonatoP09} it is conjectured that a graph $H\sim\mathcal G(k,p)$ is a homomorphic core, with high probability, if $(1+\varepsilon)\log k/k \leq p \leq 1 - (1+\varepsilon)\log k/k$ with $\varepsilon>0$ a fixed constant.

\subsection{Linkage Capacity and Treewidth}
\label{sec:intro-linkage-capacity-vs-treewidth}

\cref{thm:dense-hard,thm:avg-dense-hard} are based on lower bounds on the linkage capacity of graphs $H$ in terms of the density of $H$.
We show that the linkage capacity can also be lower-bounded as a function of the \emph{treewidth} of $H$.
As already indicated above, \cref{thm:CYBT-sparse} is actually a corollary of a much more general theorem on large-treewidth graphs shown by Marx~\cite{Marx10}:
Assuming ETH, he proved the existence of a universal constant $\alpha > 0$ such that no fixed graph $H$ with treewidth $t$ admits an $O(n^{\alpha \cdot t / \log t})$-time algorithm for $\colsub{H}$.
To obtain \cref{thm:CYBT-sparse} from this general theorem, it suffices to choose a $k$-vertex expander graph $H$ of maximum degree $3$, since such graphs are known to have treewidth $\Omega(k)$. 

We recover this theorem by showing that the linkage capacity of a graph $H$ is lower-bounded by its treewidth, up to the ``same'' logarithmic factor that is missing in the original result by Marx~\cite{Marx10}.
That is, we show the lower bound $\gamma(H) = \Omega(t / \log t)$ for every graph $H$ of treewidth $t$. 
In this proof, we use the same approximate min-cut/max-flow theorem for multicommodity flows~\cite{DBLP:journals/siamcomp/FeigeHL08,DBLP:journals/jacm/LeightonR99} that also appears in~\cite{Marx10} as black box.
Together with \cref{thm:linkage-ETH}, this indeed recovers the more general theorem of Marx~\cite{Marx10} (also for the counting version $\ccolsub{H}$) with a more transparent proof.
Complementing the lower bound, we use a simple argument about balanced separations in low-treewidth graphs to show an upper bound of $\gamma(H) = O(t)$.
We stress that both bounds are asymptotically tight.
In particular, it can be shown that $k$-vertex expander graphs $H$ of maximum degree $3$ have linkage capacity $\gamma(H) = \Theta(k/\log k)$ and treewidth $\Theta(k)$. 

It is a major open question in parameterized complexity whether the logarithmic loss in Marx's lower bound \cite{Marx10} can be avoided for all graphs $H$.

\begin{conjecture}[You Cannot Beat Treewidth \cite{Marx10}]
    Assuming ETH, there exists a universal constant $\alpha > 0$ such that
    no fixed graph $H$ with treewidth $t$
    admits an $O(n^{\alpha \cdot t})$-time algorithm for $\colsub{H}$.
\end{conjecture}

We note that even removing the $1/\log k$ factor from Theorem \ref{thm:CYBT-sparse} would constitute a significant breakthrough.
Alon and Marx~\cite{DBLP:journals/siamdm/AlonM11} showed that current approaches cannot be used to achieve this; this is also true for our techniques.
Still, with \cref{thm:dense-hard,thm:avg-dense-hard}, we extend the scope where tight bounds for $\colsub{H}$ are known.

\section{Preliminaries}

We write $\NN = \{1,2,3,\dots\}$ for the natural numbers.
For $n \in \NN$, we write $[n] \coloneqq \{1,2,\ldots,n\}$. 
All logarithms are natural unless specified otherwise. 

\subsection{Basic Definitions}

We use standard graph notation~\cite{graph}.
Graphs are finite and undirected, and we write $uv$ for edges between $u$ and $v$.
A \emph{path} from $u$ to $v$ is a sequence $P = (u=w_0,w_1,\dots,w_\ell = v)$ of distinct vertices such that consecutive vertices are adjacent.
Slightly abusing notation, we also interpret $P$ as a path from $v$ to $u$.
For a graph $G$ and $X \subseteq V(G)$, we write $G[X]$ for the subgraph induced by $X$ and $G - X \coloneqq G[V(G) \setminus X]$ for the result of removing $X$ from $G$.

A \emph{colored graph} is a triple $G = (V,E,c)$ where $c\colon V(G) \to C$ is a (not necessarily proper) \emph{coloring} of the vertices.
We say $G$ is \emph{canonically colored} if $c$ is the identity mapping and we write $\can{G}$ for the canonically colored version of $G$.

Given a ``pattern'' graph $H$ and ``host'' graph $G$, we write $\numsub{H}{G}$ for the number of subgraphs of $G$ that are isomorphic to $H$.
If $H$ and $G$ are colored, only subgraphs that preserve the coloring are counted.
For a fixed graph $H$, the problem $\ccolsub{H}$ takes as input a colored graph $G = (V,E,c)$ with $c\colon V(G) \to V(H)$, and asks to compute $\numsub{\can{H}}{G}$, while its decision version $\colsub{H}$ asks whether $\numsub{\can{H}}{G} \geq 1$.
The uncolored versions of the problem are denoted by $\csub{H}$ and $\sub{H}$ for the counting and decision version, respectively.

To analyze the complexity of $\colsub{H}$, we rely on several tools.
First, we often use graph blowups, i.e., strong graph products $H \boxtimes K_t$ of $H$ with the complete graph $K_t$. We remark that Marx~\cite{Marx10} uses the notation $H^{(t)}$ instead of $H \boxtimes K_t$.

\begin{definition}[Blowup]
    Given a graph $H$ and an integer $t \geq 1$, the \emph{blowup graph} $H \boxtimes K_t$ contains the vertices $v^{(i)}$ for all $v\in V(H)$ and $i \in [t]$, and edges
    \[\{u^{(i)}v^{(j)} \mid uv \in E(H), \, i,j \in [t]\} \ \cup \ \{u^{(i)}u^{(j)} \mid u \in V(H), \, i \neq j \in [t]\}.\]
\end{definition}

A \emph{multigraph} $M$ is a graph that allows parallel edges with the same endpoints, but no self-loops. 
The \emph{degree} $\deg_M(v)$ of a vertex $v \in V(M)$ is the number of edges incident to $v$, taking multiplicities into account.
The \emph{average degree} of $M$ is $d(M) \coloneqq 2|E(M)|/|V(M)|$.

A matching in $M$ is set $M' \subseteq E(M)$ of pairwise vertex-disjoint edges.
Slightly abusing notation, we regularly interpret $M'$ again as a graph with edge set $M'$ and vertices for all endpoints in $M'$.
A \emph{$q$-edge coloring} of $M$ is a partition of $E(M)$ into $q$ matchings.
The \emph{edge-chromatic number} of $M$, denoted by $\chi'(M)$, is the minimum number $q$ such that $M$ admits a $q$-edge coloring. 
A theorem by Shannon \cite{Shannon49} provides an upper bound on the edge-chromatic number in terms of the maximum degree, though a looser bound with factor $2$ can be achieved by a straightforward greedy algorithm.

\begin{theorem}[{\cite{Shannon49}}]
    \label{thm:shannon}
    Every multigraph $M$ of maximum degree $\Delta$ has $\chi'(M) \leq \lfloor \frac{3}{2}\Delta \rfloor$.
\end{theorem}

\subsection{Linkages}

Our hardness proofs for $\colsub{H}$ crucially rely on linkages in graphs.

\begin{definition}[Linkage and congestion]
    Given a graph $H$ and a multigraph $M$ with vertex set $X \subseteq V(H)$, an \emph{$M$-linkage} in $H$ is a collection of paths $Q = (P_{uv})_{uv\in E(M)}$ such that $P_{uv}$ has endpoints $u$ and $v$.
    For $r \in \NN$, we say that $Q$ is \emph{$r$-congested} if, for all $w\in V(H)$, at most $r$ paths $P_{uv} \in Q$ contain $w$.
    If $r=1$, we call $Q$ an \emph{uncongested} $M$-linkage.
\end{definition}

Observe that, if $Q$ is an uncongested $M$-linkage, then $M$ is necessarily a matching.
We note that we commonly work with uncongested $M$-linkages.
More precisely, we usually require uncongested $M$-linkages in blowups of graphs $H$.
Towards this end, it is often convenient to ``project'' $M$ back to the base graph $H$.
Let $H$ be a graph and let $M$ be a multigraph with $V(M) \subseteq V(H \boxtimes K_q)$.
We define the \emph{$H$-projection} of $M$ to be the multigraph $\pi(M)$ with vertex set
$V(\pi(M)) \coloneqq \{v \mid v^{(i)} \in V(M)\}$
and edge multiset
$E(\pi(M)) \coloneqq \{\!\{ vw \mid v^{(i)}w^{(j)} \in E(M), v \neq w\}\!\}$.

\begin{lemma}
    \label{lem:congestion-vs-blowup}
    Let $H$ be a graph and let $q \in \NN$.
    Also let $M$ be a matching with $V(M) \subseteq V(H \boxtimes K_q)$.
    If there is a $q$-congested $\pi(M)$-linkage in $H$, then there is an uncongested $M$-linkage in $H \boxtimes K_{2q}$.
\end{lemma}

Note that a blowup of order $2q$ rather than $q$ is needed, since we do not allow self-loops in the projection, i.e., the projection of $M$ ignores edges contained in the same block $\{v^{(i)} \mid i \in [q]\}$.
For technical reasons, we decide to handle those edges separately at the cost of losing a factor of two.

\begin{proof}
    Let $Q$ be a $q$-congested $\pi(M)$-linkage in $H$.
    We obtain an uncongested $M$-linkage $Q'$ in $H \boxtimes K_{2q}$ as follows.
    For a vertex $w \in V(H)$, let $P_1,\dots,P_t$ (if any exists) be all paths in $Q$ that contain $w$ as an internal vertex. 
    We have $t \leq q$ by definition.
    We replace $w$ in $P_i$ with the vertex $w^{(q+i)}$ from the blowup graph $H \boxtimes K_{2q}$.
    Also, all endpoints of the paths are replaced in the natural way, i.e., if $P$ has endpoints $u$ and $v$, and $uv$ is the ``projection'' of $u^{(i)}v^{(j)}$ in $M$, then $P$ gets endpoints $u^{(i)}$ and $v^{(j)}$.
    Finally, for each edge $v^{(i)}v^{(j)} \in E(M)$, we add the path $(v^{(i)},v^{(j)})$ to $Q'$.
    By the definition of the blowup graph, the resulting collection $Q'$ is an uncongested $M$-linkage in $H \boxtimes K_{2q}$.
\end{proof}

The following example illustrates the notion of linkages and the interplay between congestion and blowups; see \Cref{fig:grid-congestion}.

\begin{figure}[t]
\centering
\begin{minipage}{0.38\linewidth}
\centering
\begin{tikzpicture}
\xdef\shapeslack{0.3}
\xdef\lxslack{0.4}
\xdef\lyslack{0.7}
\def\gridpos#1#2{({(#1-1)+\lxslack*(#2-1)},{\lyslack*(#2-1)})}
\def\gridtextpos#1#2{({(#1-1)+\lxslack*(#2-1)-0.2},{\lyslack*(#2-1)+0.2})}

\foreach \x in {0,...,5} {
    \draw [black,draw opacity=0.5] (\lxslack*\x,\lyslack*\x) -- ++(5,0);
    \draw [black,draw opacity=0.5] (\x,0) -- ++(\lxslack*5,\lyslack*5);
}

\path [draw=cbfp1,ultra thick] \gridpos{1}{1} -- \gridpos{4}{1} -- \gridpos{4}{4};
\path [draw=cbfp2,ultra thick] \gridpos{2}{2} -- \gridpos{5}{2} -- \gridpos{5}{5};
\path [draw=cbfp3,ultra thick] \gridpos{3}{3} -- \gridpos{6}{3} -- \gridpos{6}{6};

\begin{scope}[shape=circle]
\node [scale=0.5,fill=cbfp1] () at \gridpos{1}{1} {};
\node [scale=0.5,fill=cbfp1] () at \gridpos{4}{4} {};
\node [scale=0.5,fill=cbfp2] () at \gridpos{2}{2} {};
\node [scale=0.5,fill=cbfp2] () at \gridpos{5}{5} {};
\node [scale=0.5,fill=cbfp3] () at \gridpos{3}{3} {};
\node [scale=0.5,fill=cbfp3] () at \gridpos{6}{6} {};

\node [scale=0.3,fill=cbfp1] () at \gridpos{2}{1} {};
\node [scale=0.3,fill=cbfp1] () at \gridpos{3}{1} {};
\node [scale=0.3,fill=cbfp1] () at \gridpos{4}{1} {};

\node [scale=0.3,fill=cbfp2] () at \gridpos{3}{2} {};
\node [scale=0.3,fill=cbfp2] () at \gridpos{5}{2} {};
\node [scale=0.3,fill=cbfp2] () at \gridpos{5}{4} {};

\node [scale=0.3,fill=cbfp3] () at \gridpos{4}{3} {};
\node [scale=0.3,fill=cbfp3] () at \gridpos{6}{3} {};
\node [scale=0.3,fill=cbfp3] () at \gridpos{6}{4} {};
\node [scale=0.3,fill=cbfp3] () at \gridpos{6}{5} {};

\node [scale=0.3,fill=black!70!white] () at \gridpos{4}{2} {};
\node [scale=0.3,fill=black!70!white] () at \gridpos{4}{3} {};
\node [scale=0.3,fill=black!70!white] () at \gridpos{5}{3} {};
\end{scope}

\foreach \x in {1,...,6} {
\node [scale=0.9] () at \gridtextpos{\x}{\x} {$v_\x$};
}

\node [scale=0.9] () at \gridpos{3.5}{2.3} {\textcolor{cbfp2}{$P^{\rightarrow}_{2,5}$}};
\node [scale=0.9] () at \gridpos{4.5}{4.3} {\textcolor{cbfp2}{$P^{\uparrow}_{2,5}$}};

\end{tikzpicture}
\end{minipage}
\begin{minipage}{0.59\linewidth}
\centering

\begin{tikzpicture}
\xdef\lxslack{0.4}
\xdef\lyslack{0.7}
\xdef\lzslack{0.5}
\foreach \x in {0,...,4} {
\foreach \y in {0,...,5} {
    \pgfmathparse{(15-\y-(5-\x))/50} \xdef\transp{\pgfmathresult}
    \draw [black,draw opacity=\transp] (\x+\lxslack*\y,\lyslack*\y) -- ++(1,-\lzslack);
    \draw [black,draw opacity=\transp] (\x+\lxslack*\y,\lyslack*\y-\lzslack) -- ++(1,\lzslack);
}
}
\foreach \x in {0,...,5} {
\foreach \y in {0,...,4} {
    \pgfmathparse{(15-\y-(5-\x))/50} \xdef\transp{\pgfmathresult}
    \draw [black,draw opacity=\transp] (\x+\lxslack*\y,\lyslack*\y) -- ++(\lxslack,\lyslack-\lzslack);
    \draw [black,draw opacity=\transp] (\x+\lxslack*\y,\lyslack*\y-\lzslack) -- ++(\lxslack,\lyslack+\lzslack);
}
}
\foreach \x in {0,...,5} {
\foreach \y in {0,...,5} {
    \draw [black,draw opacity=0.2] (\x+\lxslack*\y,\lyslack*\y) -- ++(0,-\lzslack);
}
}
\foreach \x in {0,...,5} {
    \draw [black,draw opacity=0.5] (\lxslack*\x,\lyslack*\x) -- ++(5,0);
    \draw [black,draw opacity=0.5] (\x,0) -- ++(\lxslack*5,\lyslack*5);
    \draw [black,draw opacity=0.35] (\lxslack*\x,\lyslack*\x-\lzslack) -- ++(5,0);
    \draw [black,draw opacity=0.35] (\x,-\lzslack) -- ++(\lxslack*5,\lyslack*5);
}

\def\gridpos#1#2#3{({(#1-1)+\lxslack*(#2-1)},{\lyslack*(#2-1)-#3*\lzslack})}
\def\gridtextpos#1#2#3{({(#1-1)+\lxslack*(#2-1)-0.2},{\lyslack*(#2-1)-#3*\lzslack+0.2})}

\foreach \x in {1,...,6} {
\node [scale=0.9] () at \gridtextpos{\x}{\x}{0} {$v_\x^{(1)}$};
}
\node [scale=0.9] () at \gridtextpos{1}{1}{1.2} {$v_1^{(2)}$};
\node [scale=0.3,fill=black,shape=circle] () at \gridpos{1}{1}{1} {};

\draw [ultra thick,cbfp3!30!white] \gridpos{3}{3}{0} -- \gridpos{4}{3}{1};
\draw [ultra thick,cbfp3, path fading=south] \gridpos{3}{3}{0} -- \gridpos{4}{3}{1};
\draw [ultra thick,cbfp3!30!white] \gridpos{4}{3}{1} -- \gridpos{5}{3}{0};
\draw [ultra thick,cbfp3, path fading=south] \gridpos{4}{3}{1} -- \gridpos{5}{3}{0};
\draw [ultra thick,cbfp1!30!white] \gridpos{4}{1}{0} -- \gridpos{4}{2}{1};
\draw [ultra thick,cbfp1, path fading=east] \gridpos{4}{1}{0} -- \gridpos{4}{2}{1};
\draw [ultra thick,cbfp1!30!white] \gridpos{4}{2}{1} -- \gridpos{4}{3}{0};
\draw [ultra thick,cbfp1, path fading=south] \gridpos{4}{1.92}{0.92} -- \gridpos{4}{2}{1} -- \gridpos{4}{3}{0} ;
\draw [ultra thick,cbfp2!30!white] \gridpos{5}{3}{1} -- \gridpos{5}{4}{0};
\draw [ultra thick,cbfp2, path fading=south] \gridpos{5}{2.92}{0.92} -- \gridpos{5}{3}{1} -- \gridpos{5}{4}{0};
\draw [ultra thick,cbfp2!30!white] \gridpos{5}{2}{0} -- \gridpos{5}{3}{1};
\draw [ultra thick,cbfp2, path fading=east] \gridpos{5}{2}{0} -- \gridpos{5}{3}{1};
\draw [ultra thick,cbfp1] \gridpos{1}{1}{0} -- \gridpos{4}{1}{0};
\draw [ultra thick,cbfp1] \gridpos{4}{3}{0} -- \gridpos{4}{4}{0};
\draw [ultra thick,cbfp2] \gridpos{2}{2}{0} -- \gridpos{5}{2}{0};
\draw [ultra thick,cbfp2] \gridpos{5}{4}{0} -- \gridpos{5}{5}{0};
\draw [ultra thick,cbfp3] \gridpos{5}{3}{0} -- \gridpos{6}{3}{0};
\draw [ultra thick,cbfp3] \gridpos{6}{3}{0} -- \gridpos{6}{6}{0};

\draw [black,draw opacity=0.5] \gridpos{5}{2}{0} -- \gridpos{6}{2}{0}; 
\draw [black,draw opacity=0.5] \gridpos{4}{1}{0} -- \gridpos{5}{1}{0}; 

\begin{scope}[shape=circle]
\node [scale=0.5,fill=cbfp1] () at \gridpos{1}{1}{0} {};
\node [scale=0.5,fill=cbfp1] () at \gridpos{4}{4}{0} {};
\node [scale=0.5,fill=cbfp2] () at \gridpos{2}{2}{0} {};
\node [scale=0.5,fill=cbfp2] () at \gridpos{5}{5}{0} {};
\node [scale=0.5,fill=cbfp3] () at \gridpos{3}{3}{0} {};
\node [scale=0.5,fill=cbfp3] () at \gridpos{6}{6}{0} {};

\node [scale=0.3,fill=cbfp1] () at \gridpos{2}{1}{0} {};
\node [scale=0.3,fill=cbfp1] () at \gridpos{3}{1}{0} {};
\node [scale=0.3,fill=cbfp1] () at \gridpos{4}{1}{0} {};
\node [scale=0.3,fill=cbfp1!30!white] () at \gridpos{4}{2}{1} {};
\node [scale=0.3,fill=cbfp1] () at \gridpos{4}{3}{0} {};

\node [scale=0.3,fill=cbfp2] () at \gridpos{3}{2}{0} {};
\node [scale=0.3,fill=cbfp2] () at \gridpos{4}{2}{0} {};
\node [scale=0.3,fill=cbfp2] () at \gridpos{5}{2}{0} {};
\node [scale=0.3,fill=cbfp2!30!white] () at \gridpos{5}{3}{1} {};
\node [scale=0.3,fill=cbfp2] () at \gridpos{5}{4}{0} {};

\node [scale=0.3,fill=cbfp3!30!white] () at \gridpos{4}{3}{1} {};
\node [scale=0.3,fill=cbfp3] () at \gridpos{5}{3}{0} {};
\node [scale=0.3,fill=cbfp3] () at \gridpos{6}{3}{0} {};
\node [scale=0.3,fill=cbfp3] () at \gridpos{6}{4}{0} {};
\node [scale=0.3,fill=cbfp3] () at \gridpos{6}{5}{0} {};
\end{scope}

\end{tikzpicture}
\end{minipage}

\vspace*{0.5em}

\begin{minipage}{0.44\linewidth}
\centering
(a)
\end{minipage}
\begin{minipage}{0.50\linewidth}
\centering
(b)
\end{minipage}
\caption{(a) The grid graph $\boxplus_{6}$. Thick paths depict a $2$-congested $M$-linkage, where $M=\{\textcolor{cbfp1}{v_1v_4},\textcolor{cbfp2}{v_2v_5},\textcolor{cbfp3}{v_3v_6}\}$ is a matching on the diagonal vertices.
(b) The blowup graph $\boxplus_{6}\boxtimes K_2$, and an uncongested $M$-linkage obtained from the $2$-congested $M$-linkage in $\boxplus_{6}$.
}  
\label{fig:grid-congestion}
\end{figure}

\begin{example}[Grid graph] \label{exp:grid}
Write $\boxplus_\ell$ for the grid graph on vertex set $[\ell] \times [\ell]$. 
For every matching $M$ on the set of diagonal vertices $v_i =(i,i)$ for $i \in [\ell]$, we observe that $\boxplus_\ell$ contains a $2$-congested $M$-linkage $Q(M) = \{P_{uv}\}_{uv \in M}$. This $2$-congested linkage induces an uncongested $M$-linkage in $\boxplus_\ell \boxtimes K_2$ via Lemma~\ref{lem:congestion-vs-blowup}.
See also \Cref{fig:grid-congestion}.

More specifically, given an edge from $u=(a,a)$ to $v =(b,b)$ for $a<b$, we define $P_{uv}$ as the concatenation of the path $P^\rightarrow_{uv}$ on vertices $u=(a,a), \ldots, (b,a)$ and the path $P^\uparrow_{uv}$ on vertices $(b,a), \ldots , (b,b)=v$.
The paths $P^\rightarrow_{uv}$ for $uv \in M$ are vertex-disjoint (as distinct paths have distinct $y$-coordinates), and so are the paths $P^\uparrow_{uv}$ for $uv \in M$ (having distinct $x$-coordinates), so $Q(M)$ is indeed a $2$-congested $M$-linkage. 
\end{example}

\section{Lower Bounds from Linkage Capacity}

The Exponential-Time Hypothesis ETH~\cite{IP01} postulates the existence of a constant $\alpha > 0$ such that no $O(2^{\alpha \cdot n})$ time algorithm can decide, on input a $3$-CNF formula $\varphi$ with $n$ variables, whether $\varphi$ admits a satisfying assignment. 
Its \emph{a priori} weaker counting version \#ETH postulates the same lower bound for \emph{counting} the satisfying assignments of $\varphi$~\cite{DBLP:journals/talg/DellHMTW14}. 
For both hypotheses, the sparsification lemma \cite{IPZ01, DBLP:journals/talg/DellHMTW14} rules out such algorithms even under the additional condition that every variable in $\varphi$ appears in at most $C$ clauses, for some constant $C \in \NN$.
By a standard reduction, lower bounds follow for the problem $\threecol$ of deciding whether an input graph $G$ admits a proper vertex-coloring with $3$ colors where no adjacent vertices receive the same color; see for example \cite[Theorem 3.2]{LMS11}.

\begin{theorem}
    \label{thm:3col-hard}
    Assuming ETH, there is a constant $\beta > 0$
    such that $\threecol$ cannot be solved in time $O(2^{\beta \cdot n})$ for $n$-vertex input graphs $G$ of maximum degree $4$.
    The same holds for $\countthreecol$ under \#ETH.
\end{theorem}

This theorem is the foundation for the lower bounds shown in this paper.

\subsection{Instances That Fit into Blowups}

It is useful for us to generalize $\threecol$ slightly, by allowing edges to either enforce equality or disequality of their endpoint colors.
Since ``equality edges'' can be contracted without changing the number of valid assignments,
we obtain an immediate way to simulate edges in a $\threecol$ instance by paths.

\begin{definition}
    Given a graph $G=(V,E)$ with a partition $E = E_= \cup E_{\neq}$, a \emph{proper $3$-assignment} is a function $a : V \to [3]$ such that $a(u) = a(v)$ for all $uv \in E_=$, while $a(u) \neq a(v)$ for all $uv \in E_{\neq}$.
    The problem $\threeass$ asks to determine the existence of a proper assignment on input $(G, E_=, E_{\neq})$, while $\countthreeass$ asks to count them.
\end{definition}

It is possible to convert instances $G$ for $\threeass$ into instances $X$ for $\colsub{H}$.
Moreover, if $G$ fits into a moderately small blowup of $H$, then $X$ is only of moderately exponential size.
This can be shown by a simple ``split-and-list'' reduction that follows the sketch given in the introduction for $\pcliquedecision$.

\begin{lemma}
    \label{lem:split-list}
    Let $H$ be a fixed $k$-vertex graph with canonical vertex-coloring.
    Given a subgraph $G$ of $H \boxtimes K_t$ for $t \in \NN$, a colored graph $X$ on $k \cdot 3^t$ vertices can be computed in $9^t \cdot \mathrm{poly}(k,t)$ time such that $\numsub{H}{X}
    $ equals the number of proper $3$-assignments in $G$.
\end{lemma}

\begin{proof}
    Suppose $V(H) = [k]$ and consider the partition of the vertex set $V(G)$ into the sets $V_w \coloneqq \{w^{(1)},\dots,w^{(t)}\} \cap V(G)$ for $w \in [k]$.
    Define $X_w$ for $w \in [k]$ as the set of all proper $3$-assignments to $G[V_w]$.
    Observe that $V_w$ may be empty; in this case $X_w$ contains exactly the empty mapping.

    For $w,w' \in [k]$, we call two $3$-assignments $a \in X_w$ and $a' \in X_{w'}$ \emph{compatible} if their union is a proper $3$-assignment of $G[V_w \cup V_{w'}]$.
    Let us define
    \begin{align*}
        A_{X} &\coloneqq \{(a_{1},\ldots,a_{k})\in X_{1}\times\ldots\times X_{k} \mid \forall ww'\in E(H):a_{w}\text{ and }a_{w'}\text{ are compatible}\}\\
        A_{G} &\coloneqq \{a\colon V(G)\to[3] \mid a\text{ is proper }3\text{-assignment of }G\}
    \end{align*}
    
    We observe that the map $a \mapsto (a_1, \ldots, a_k)$ from $A_G$ to $A_X$, where $a_w$ is the restriction of $a$ to $V_w$, is a bijection. 
    Indeed, in the image of $a \in A_G$ under this map, $a_w$ and $a_{w'}$ are compatible for all $w,w' \in E(H)$.
    Conversely, given $(a_1, \ldots, a_k) \in A_X$, recall that every edge $uv\in E(G)$ satisfies $u\in V_w$ and $v \in V_{w'}$ for some $w,w' \in [k]$ with (a) $w=w'$, or (b) $ww' \in E(H)$.
    In case (a), since $a_w$ is proper, the endpoints of $uv$ receive a proper assignment under $a$.
    In case (b), because the union of $a_w$ and $a_{w'}$ is a proper $3$-assignment of $G[V_w \cup V_{w'}]$, the endpoints of $uv$ receive a proper assignment under $a$.
    Thus $a \in A_G$.

    Finally, the graph $X$ is defined on vertices $\bigcup_{w \in [k]} X_w$, where each vertex in $X_w$ is colored by $w \in [k]$.
    An edge is present between $a \in X_w$ and $a' \in X_{w'}$ if and only if $ww' \in E(H)$ and $a$ and $a'$ are compatible.
    The (colored) subgraphs $S$ of $X$ isomorphic to $H$ correspond to tuples in $A_X$.
    Indeed, $V(S)$ corresponds to a tuple $(a_{1},\dots,a_{k})\in X_{1}\times\dots\times X_{k}$, and the presence of edges of $H$ in $S$ implies that $a_w$ and $a_{w'}$ are compatible for $ww' \in E(H)$.
    Since $|X_w| \leq 3^t$ for all $w \in [k]$, the graph $X$ can be computed by brute-force in $9^t \cdot \mathrm{poly}(k,t)$ time.
\end{proof}

\subsection{The Linkage Capacity of a Graph}

Next, we formally define the notion of \emph{linkage capacity}.
First, we need to define a term for vertex sets $X$ in graphs that can be paired up arbitrarily via paths in $H$.
This resembles Diestel's~\cite{graph} notion of \emph{linkedness}, see~\Cref{sec:dense-worstcase}, which however requires this property for the entire graph $H$.

\begin{definition}[Matching-linked set \& linkage number]
    Given a graph $H$, we say that $X \subseteq V(H)$ is \emph{matching-linked} if $H$ contains an uncongested $M$-linkage for every matching $M$ on vertex set $X$.
    The \emph{linkage number} of $H$, denoted by $\ell(H)$, is the maximum size of a matching-linked set.
\end{definition}

\begin{remark}
    Let $X \subseteq V(H)$ be a matching-linked set in a graph $H$.
    We stress that for \emph{every} matching $M$ on vertex set $X$ there is an uncongested $M$-linkage in $H$.
    In particular, the edges of $M$ are \emph{not} necessarily  contained in $E(H)$: Only the \emph{endpoints} of $M$ need to be contained in $X$.
\end{remark}

A simple edge-coloring argument, also used in \Cref{lem:reroute-G-to-G'}, shows that large matching-linked sets $X$ in blowups $H \boxtimes K_t$ suffice to embed graphs $G$ of maximum degree $\Delta$ into $H \boxtimes K_{2\Delta \cdot t}$.
Thus, large matching-linked sets $X$ \emph{in blowups of $H$} are a useful ``resource'' attainable from $H$ that allows us to use \Cref{lem:split-list}.

Not all such sets $X$ however need to originate from matching-linked sets in $H$ itself.
Consider a set $X$ in $H$ that \emph{just} fails to be matching-linked, as in \Cref{exp:grid}, in the sense that $X$ still admits $M$-linkages of congestion $2$ in $H$.
Such $M$-linkages then induce uncongested $M$-linkages in $H \boxtimes K_2$.
As our goal is to embed a $\threecol$ instance $G$ into a moderately large blowup of $H$, such a constant-factor loss would be acceptable.
This flexibility is captured by the \emph{linkage capacity}, which measures the maximum size of matching-linked sets in blowups of $H$ relative to the blowup order.

\begin{definition}[Linkage capacity]\label{def:linkage:capacity}
    The \emph{linkage capacity} of a graph $H$ is
    \[\gamma(H) \coloneqq \liminf_{t \to \infty} \frac{\ell(H\boxtimes K_t)}{t}.\]
\end{definition}

For every graph $H$ it trivially holds that $1 \leq \ell(H\boxtimes K_t)/t \leq |V(H)|$ for all $t \in \NN$.
In particular, $1 \leq \gamma(H) \leq |V(H)|$.
By taking the limes inferior over this sequence of numbers, we ensure that large matching-linked sets are available in \emph{all} sufficiently large blowups of $H$. 
The following simple example demonstrates the necessity to consider the limit.

\begin{example}
    Let $P_2$ denote the $3$-vertex path.
    It can be verified that
    \[\frac{\ell(P_2 \boxtimes K_t)}{t} = 2 + \frac{1}{t}\]
    for all $t \in \NN$.
    In particular, $\gamma(P_2) = 2$ although $P_2$ contains a matching-linked set of size $3$.
\end{example}

The example also highlights that taking the supremum over the sequence $\ell(H\boxtimes K_t)/t$ (which is a viable alternative definition) leads to a different graph parameter.
However, the next lemma shows that both options are within a constant factor of one another.
In particular, large matching-linked sets in $H$, or small blowups of $H$, are sufficient to establish high linkage capacity.

\begin{lemma}
    \label{lem:linkage-capacity-via-blowup}
    Let $H$ be a graph.
    Then, for every $q \in \NN$,
    \[\gamma(H) \geq \frac{1}{3} \cdot \frac{\ell(H\boxtimes K_q)}{q}.\]
\end{lemma}

\begin{proof}
    The proof rests on the following claim.
    \begin{claim}
        \label{claim:matching-linked-set-in-blowup}
        Let $H'$ be a graph and suppose $X' \subseteq V(H')$ is a matching-linked set.
        Then $X_t' \coloneqq \{v^{(i)} \mid v \in X', 1 \leq i \leq t/3\}$ is matching-linked in $H' \boxtimes K_t$ for every $t \in \NN$.
    \end{claim}
    \begin{claimproof}
        Let $M$ be a matching on $X_t'$ and let $M' \coloneqq \pi(M)$ be the $H'$-projection of $M$.
       Observe that $M'$ has maximum degree at most $t/3$, so $\chi'(M') \leq t/2$ by \Cref{thm:shannon}.
        Hence, the multigraph $M$ can be partitioned into $r \leq t/2$ matchings $M_1,\dots M_r$ on $X'$.
        Since $X'$ is a matching-linked set, for every $M_i$, there is an uncongested $M_i$-linkage $Q_i$ in $H'$.
        So
        $Q \coloneqq \bigcup_{i \in [r]}Q_i$
        is a $r$-congested $M'$-linkage in $H'$, and there is an uncongested $M$-linkage in $H' \boxtimes K_t$ by \Cref{lem:congestion-vs-blowup}.
    \end{claimproof}
    
    Now, let $X \subseteq V(H\boxtimes K_q)$ be a matching-linked set in $H\boxtimes K_q$ of size $|X| = \ell(H\boxtimes K_q)$.
    Let $c < \frac{1}{3} \cdot |X|/q$.
    We define
    \[t_0 \coloneqq \left(\frac{1}{3} \cdot \frac{|X|}{q} - c\right)^{-1} \cdot |X|.\]
    We argue that
    \begin{equation}
        \label{eq:linkage-inf}
        \ell(H\boxtimes K_t) \geq c \cdot t
    \end{equation}
    for all $t \geq t_0$.
    Indeed, consider the graph $H' \coloneqq H \boxtimes K_q$ and set $\ell \coloneqq \lfloor t/q \rfloor$.
    By Claim \ref{claim:matching-linked-set-in-blowup} the graph $H' \boxtimes K_\ell$ contains a matching-linked set $X_\ell'$ with
    \[
        |X_\ell'| \geq |X| \cdot \left\lfloor \frac{\ell}{3} \right\rfloor \geq |X| \cdot \left(\frac{\ell}{3} - \frac{2}{3}\right)
        \geq |X| \cdot \left(\frac{1}{3} \cdot \left(\frac{t}{q}-1\right) - \frac{2}{3}\right)
        = \frac{1}{3} \cdot \frac{|X|}{q} \cdot t - |X| \geq c\cdot t.
    \]
    Since $(H \boxtimes K_q) \boxtimes K_\ell$ is a subgraph of $H \boxtimes K_t$, we obtain Equation \eqref{eq:linkage-inf} for all $t \geq t_0$.
    Because $c < \frac{1}{3} \cdot |X|/q$ was chosen arbitrarily, the lemma follows.
\end{proof}

As a concrete example, let us use \cref{exp:grid} to bound the linkage capacity of grids.

\begin{lemma}
    \label{lem:linkage-capacity-grid}
    For the $\ell$-by-$\ell$ grid graph $\boxplus_\ell$, we have $\gamma(\boxplus_\ell)\geq \ell/6$. 
\end{lemma}

\begin{proof}
    Suppose $V(\boxplus_\ell) = [\ell]^2$.
    By \cref{exp:grid}, the set $X \coloneqq \{(i,i)^{(1)} \mid i\in [\ell]\}$ is matching-linked in the blowup $\boxplus_\ell\boxtimes K_2$. 
    The lemma follows by invoking \cref{lem:linkage-capacity-via-blowup}. 
\end{proof}

The next lemma shows that linkage capacity has the useful property of being minor-closed.
A graph $H'$ is a \emph{minor} of another graph $H$ if there are pairwise disjoint sets $B(v') \subseteq V(H)$, $v' \in V(H')$, such that
\begin{itemize}
    \item $B(v')$ induces a connected subgraph in $H$ for every $v' \in V(H')$, and
    \item for every $u'v' \in E(H')$ there are $u \in B(u')$, $v \in B(v')$ such that $uv \in E(H)$.
\end{itemize}
We call the sets $B(v')$, $v' \in V(H')$, a \emph{minor model} of $H'$ in $H$.

\begin{lemma}\label{lem:minor-closed}
    Let $H'$ be minor of $H$. Then $\ell(H') \leq \ell(H)$ and $\gamma(H') \leq \gamma(H)$.
\end{lemma}
\begin{proof}
    We first show $\ell(H') \leq \ell(H)$.
    Let $B(v')$, $v' \in V(H')$, be a minor model of $H'$ in $H$.
    For each $v' \in V(H')$, we fix an arbitrary vertex $v \in B(v')$ that represents that vertex.\footnote{Formally, we define a mapping $f(v') = v$. We implicitly perform this mapping by removing the ${\phantom{v}}'$-symbol.} 
    
    Let $X' = \{v_1',\dots,v_s'\}$ be a maximum matching-linked set in $H'$.
    We show that $X \coloneqq \{v_1,\dots,v_s\}$ is a matching-linked set in $H$.
    Towards this end, let $M$ be a matching on $X$. 
    Now, $M$ induces a matching $M'$ of $X'$ (by identifying a vertex $v$ with $v'$).
    Since $X'$ is a matching-linked set, there exists an uncongested $M'$-linkage $Q' = (P'_{u'v'})_{u'v' \in M'}$ in $H'$.
    This collection naturally translates this into an uncongested $M$-linkage in $H$:
    For a path $P'_{u'v'}$ let $B(P'_{u'v'}) \coloneqq \bigcup_{w' \in V(P'_{u'v'})} B(w')$.
    Then $B(P'_{u'v'})$ induces a connected subgraph of $H$ and $u,v \in B(P'_{u'v'})$.
    So there is a path $P_{uv}$ from $u$ to $v$ in $H$ with all internal vertices contained in $B(P'_{u'v'})$.
    The resulting collection of paths $Q = (P_{uv})_{uv \in M}$ is an $M$-linkage in $H$.
    So $\ell(H) \geq |X| = |X'| = \ell(H')$.

    Next, we show that  $\gamma(H') \leq \gamma(H)$.
    Towards this end, observe that $H' \boxtimes K_t$ is a minor of $H \boxtimes K_t$ for every $t \in \NN$.
    Indeed, let $B(\hat{v})$, $\hat{v} \in V(H')$, be a minor model of $H'$ in $H$.
    Then $B_t(\hat{v}^{(i)}) \coloneqq \{w^{(i)} \mid w \in B(\hat{v})\}$, $\hat{v} \in V(H')$, $i \in [t]$, is a minor model of $H' \boxtimes K_t$ in $H \boxtimes K_t$.
    We conclude that $\ell(H' \boxtimes K_t) \leq \ell(H \boxtimes K_t)$ for every $t \in \NN$.
    It follows that $\gamma(H') \leq \gamma(H)$.
\end{proof}

\subsection{Fitting Instances into Blowups via Linkage Capacity}

Having introduced linkage capacity and its key properties, we now use it to embed graphs $G$ into blowups $H \boxtimes K_t$ with $t = O(n / \gamma(H))$.
If we can show that $\gamma(H)$ is large, then ETH implies strong lower bounds for $\colsub{H}$ via \Cref{thm:3col-hard} and \Cref{lem:split-list}.

A minor constructivity issue arises:
Some techniques for lower-bounding $\gamma(H)$ do not necessarily yield efficient algorithms for finding linkages in blowups of $H$.
Thus, it could \emph{a priori} be possible for $G$ to embed into $H \boxtimes K_t$, yet we cannot efficiently find an embedding.
This concern is resolved by known algorithms for graphs of bounded neighborhood diversity~\cite[Theorem~3.7]{Ganian12}, or via the self-contained proof provided in the appendix:

\begin{theorem}\label{theo:find-linked-set}
    Let $f(k) = 3 k^{k + 2}$.
    Given a $k$-vertex graph $H$ and $t \geq 2$ as input, a matching-linked set $X$ of maximum size in $H \boxtimes K_t$ can be found in $O(t^{f(k)})$ time.
    Given additionally a matching-linked set $X$ in $H \boxtimes K_t$ and a matching $M$ with vertex set $X$, an $M$-linkage in $H \boxtimes K_t$ can be found in $O(t^{f(k)})$ time.
\end{theorem}

We can now turn to our main part of the hardness proof.
In the following, given graphs $G$ and $G'$ without loops or multi-edges, a $G$-linkage $Q = (P_{uv})_{uv\in E(G)}$ in $G'$ is a \emph{topological $G$-minor model} in $G'$ if paths $P_{uv}$ and $P_{u'v'}$ for $uv,u'v' \in E(G)$ in $Q$ intersect only at endpoints. In particular, such intersections can occur only if $uv$ and $u'v'$ share a common vertex.
We refer to the subgraph of $G'$ induced by $Q$ as the \emph{image} of $Q$.

\begin{lemma}
    \label{lem:reroute-G-to-G'}
    Let $H$ be a fixed $k$-vertex graph.
    Then there is $n_0 \in \NN$ such that for every graph $G$ with $n \geq n_0$ vertices and maximum degree $\Delta$, it holds that $H \boxtimes K_t$ contains a topological $G$-minor for every
    \[t \geq 3\Delta \cdot \frac{n}{\gamma(H)}.\]

    Moreover, setting $f(k) = 3 k^{k + 2}$, there is an $O(n^{f(k)})$ time algorithm that, given a graph $G$ with $n \geq n_0$ vertices and maximum degree $\Delta$, computes a topological $G$-minor model in $H \boxtimes K_{t^*}$ where $t^* \coloneqq \lceil 3\Delta \cdot n/\gamma(H) \rceil$.
\end{lemma}

\begin{figure}[t]
\centering
\begin{minipage}{0.46\linewidth}
\centering
\begin{tikzpicture}[xscale=2.2,yscale=2.2]
\xdef\lxslack{0.4}
\xdef\lyslack{0.7}
\def\gridpos#1#2{({(#1-1)+\lxslack*(#2-1)},{\lyslack*(#2-1)})}
\def\gridtextpos#1#2{({(#1-1)+\lxslack*(#2-1)-0.2},{\lyslack*(#2-1)+0.2})}

\node [shape=circle,fill=black,scale=0.5] (v1) at \gridpos{1}{3} {};   
\node [shape=circle,fill=black,scale=0.5] (v2) at \gridpos{2}{3} {};   
\node [shape=circle,fill=black,scale=0.5] (v3) at \gridpos{3}{3} {};   
\node [shape=circle,fill=black,scale=0.5] (v4) at \gridpos{3.2}{2.6} {};   
\node [shape=circle,fill=black,scale=0.5] (v5) at \gridpos{1}{2} {};   
\node [shape=circle,fill=black,scale=0.5] (v6) at \gridpos{2}{2} {};   
\node [shape=circle,fill=black,scale=0.5] (v7) at \gridpos{3}{2} {};   
\node [shape=circle,fill=black,scale=0.5] (v8) at \gridpos{1}{1} {};   
\node [shape=circle,fill=black,scale=0.5] (v9) at \gridpos{1.3}{0.7} {};   
\node [shape=circle,fill=black,scale=0.5] (v10) at \gridpos{2}{1} {};   
\node [shape=circle,fill=black,scale=0.5] (v11) at \gridpos{2.3}{0.7} {};   
\node [shape=circle,fill=black,scale=0.5] (v12) at \gridpos{3}{1} {};   
\node [shape=circle,fill=black,scale=0.5] (v13) at \gridpos{3.3}{0.7} {};

\path [draw=black!70!white] (v1) -- (v2) -- (v3) -- (v7) -- (v4) -- (v3);
\path [draw=black!70!white] (v1) -- (v5) -- (v6) -- (v2);
\path [draw=black!70!white] (v6) -- (v7) -- (v12) -- (v13) -- (v7);
\path [draw=black!70!white] (v12) -- (v11) -- (v9) -- (v8) -- (v5) -- (v6) -- (v10);
\path [draw=black!70!white] (v8) -- (v11);

\path [draw=cbfp1, ultra thick] (v1) edge [bend right] node {} (v8);
\path [draw=cbfp1, ultra thick] (v3) -- (v6);
\path [draw=cbfp2, ultra thick] (v2) edge [bend right=15] node {} (v8);
\path [draw=cbfp2, ultra thick] (v4) edge [bend right=15] node {} (v11);
\path [draw=cbfp3, ultra thick] (v1) -- (v10);
\path [draw=cbfp3, ultra thick] (v2) -- (v7);
\end{tikzpicture}
\end{minipage}
\begin{minipage}{0.52\linewidth}
\centering
\begin{tikzpicture}[xscale=2.2,yscale=2.2]
\xdef\lxslack{0.4}
\xdef\lyslack{0.7}
\xdef\lzslack{0.15}
\def\gridpos#1#2#3{({(#1-1)+\lxslack*(#2)-sin(#3*120-75)*\lzslack},{\lyslack*(#2-1)+cos(#3*120-75)*\lzslack})}
\def\gridcentrepos#1#2{({(#1-1)+\lxslack*(#2)},{\lyslack*(#2-1)})}

\foreach \x in {1,2}{
\foreach \y in {1,2}{
\foreach \tt in {0,1,2}{
\foreach \ttt in {0,1,2}{
\path [draw=black!16!white] \gridpos{\x}{\y}{\tt} -- \gridpos{\x+1}{\y}{\ttt};
\path [draw=black!16!white] \gridpos{\x}{\y}{\tt} -- \gridpos{\x}{\y+1}{\ttt};
}}}}
\foreach \x in {1,2}{
\foreach \y in {3}{
\foreach \tt in {0,1,2}{
\foreach \ttt in {0,1,2}{
\path [draw=black!16!white] \gridpos{\x}{\y}{\tt} -- \gridpos{\x+1}{\y}{\ttt};
}}}}
\foreach \x in {3}{
\foreach \y in {1,2}{
\foreach \tt in {0,1,2}{
\foreach \ttt in {0,1,2}{
\path [draw=black!16!white] \gridpos{\x}{\y}{\tt} -- \gridpos{\x}{\y+1}{\ttt};
}}}}
\foreach \x in {1,2,3}{
\foreach \y in {1,2,3}{
\path [draw=black!16!white] \gridpos{\x}{\y}{0} -- \gridpos{\x}{\y}{1} -- \gridpos{\x}{\y}{2} -- \gridpos{\x}{\y}{0};
}
}

\node [shape=circle,fill=black,scale=0.5] (v1) at \gridpos{1}{3}{0} {};   
\node [shape=circle,fill=black,scale=0.5] (v2) at \gridpos{2}{3}{2} {};   
\node [shape=circle,fill=black,scale=0.5] (v3) at \gridpos{3}{3}{1} {};   
\node [shape=circle,fill=black,scale=0.5] (v4) at \gridpos{3}{3}{2} {};   
\node [shape=circle,fill=black,scale=0.5] (v5) at \gridpos{1}{2}{2} {};   
\node [shape=circle,fill=black,scale=0.5] (v6) at \gridpos{2}{2}{2} {};   
\node [shape=circle,fill=black,scale=0.5] (v7) at \gridpos{3}{2}{0} {};   
\node [shape=circle,fill=black,scale=0.5] (v8) at \gridpos{1}{1}{0} {};   
\node [shape=circle,fill=black,scale=0.5] (v9) at \gridpos{1}{1}{2} {};   
\node [shape=circle,fill=black,scale=0.5] (v10) at \gridpos{2}{1}{1} {};   
\node [shape=circle,fill=black,scale=0.5] (v11) at \gridpos{2}{1}{2} {};   
\node [shape=circle,fill=black,scale=0.5] (v12) at \gridpos{3}{1}{2} {};   
\node [shape=circle,fill=black,scale=0.5] (v13) at \gridpos{3}{1}{0} {};

\node [shape=circle,fill=cbfp1,scale=0.4] (r1) at \gridpos{1}{2}{1} {};   
\node [shape=circle,fill=cbfp1,scale=0.4] (r2) at \gridpos{2}{3}{1} {};   
\node [shape=circle,fill=cbfp2,scale=0.4] (r3) at \gridpos{2}{2}{0} {};   
\node [shape=circle,fill=cbfp2,scale=0.4] (r4) at \gridpos{2}{1}{0} {};   
\node [shape=circle,fill=cbfp2,scale=0.4] (r5) at \gridpos{3}{2}{1} {};   
\node [shape=circle,fill=cbfp2,scale=0.4] (r6) at \gridpos{3}{1}{1} {};   
\node [shape=circle,fill=cbfp3,scale=0.4] (r7) at \gridpos{1}{2}{0} {};   
\node [shape=circle,fill=cbfp3,scale=0.4] (r8) at \gridpos{2}{2}{1} {};   
\node [shape=circle,fill=cbfp3,scale=0.4] (r9) at \gridpos{3}{3}{0} {};   

\node [shape=circle,fill=black!30!white,scale=0.3] () at \gridpos{1}{1}{1} {};
\node [shape=circle,fill=black!30!white,scale=0.3] () at \gridpos{3}{2}{2} {};
\node [shape=circle,fill=black!30!white,scale=0.3] () at \gridpos{1}{3}{1} {};
\node [shape=circle,fill=black!30!white,scale=0.3] () at \gridpos{1}{3}{2} {};
\node [shape=circle,fill=black!30!white,scale=0.3] () at \gridpos{2}{3}{0} {};

\path [draw=black!70!white] (v1) -- (v2) -- (v3) -- (v7) -- (v4) -- (v3);
\path [draw=black!70!white] (v1) -- (v5) -- (v6) -- (v2);
\path [draw=black!70!white] (v6) -- (v7) -- (v12) -- (v13) -- (v7);
\path [draw=black!70!white] (v12) -- (v11) -- (v9) -- (v8) -- (v5) -- (v6) -- (v10);
\path [draw=black!70!white] (v8) -- (v11);

\path [draw=cbfp1,ultra thick] (v1) -- (r1) -- (v8);
\path [draw=cbfp1,ultra thick] (v3) -- (r2) -- (v6);
\path [draw=cbfp2,ultra thick] (v2) -- (r3) -- (r4) -- (v8);
\path [draw=cbfp2,ultra thick] (v4) -- (r5) -- (r6) -- (v11);
\path [draw=cbfp3,ultra thick] (v1) -- (r7) -- (r8) -- (v10);
\path [draw=cbfp3,ultra thick] (v2) -- (r9)  -- (v7);

\end{tikzpicture}
\end{minipage}

\vspace*{0.5em}

\begin{minipage}{0.44\linewidth}
\centering
(a)
\end{minipage}
\begin{minipage}{0.50\linewidth}
\centering
(b)
\end{minipage}
\caption{(a) A graph $G$ that fails to be embedded into the blowup $\boxplus_{3} \boxtimes K_2$ due to the colored edges, which are partitioned into three matchings. (b) An embedding of $G$ into $\boxplus_{3} \boxtimes K_3$ as a topological minor, where each colored edge gets routed via new vertices from the blowup.}
\label{fig:test}
\end{figure}

\begin{proof}
    By definition of $\gamma(H)$, there is $s_H \in \NN$ such that
    \[\ell(H \boxtimes K_s) \geq \frac{3}{4} \cdot \gamma(H) \cdot s\]
    for every $s \geq s_H$.
    We fix $n_0 \in \NN$ such that $n_0 \geq \max(6,s_H) \cdot \gamma(H)$.

    Let $t' \coloneqq \lceil \frac{4}{3} \cdot n/\gamma(H) \rceil \leq \frac{3}{2} \cdot n/\gamma(H)$.
    Since $n \geq n_0$, we conclude that $t' \geq s_H$.
    Hence,
    \[\ell(H \boxtimes K_{t'}) \geq \frac{3}{4} \cdot \gamma(H) \cdot t' \geq n,\]
    i.e., the graph $H \boxtimes K_{t'}$ contains a matching-linked set $X \subseteq V(H \boxtimes K_{t'})$ of size $n$.
    In $O(n^{f(k)})$ time, \Cref{theo:find-linked-set} finds such a set $X$.
    Fix $V(G) = X$ in the following.
    
    Using the straightforward greedy algorithm, we can compute a $q$-edge-coloring $E(G) = M_1 \cup \dots \cup M_q$ of $G$ in time $O(\Delta n)$ where $q = 2 \Delta - 1$.
    Since $X$ is matching-linked, the graph $H \boxtimes K_{t'}$ contains an uncongested $M_i$-linkage $Q_i$ for every individual $i \in [q]$.
    Each linkage can be found $O(n^{f(k)})$ time via \Cref{theo:find-linked-set}.
    These linkages together induce a topological $G$-minor model in $H \boxtimes K_{qt'}$ (see Figure \ref{fig:test}):
    Consider $V(H \boxtimes K_{qt'})$ to be partitioned into $q$ layers such that layer $i \in [q]$ contains the vertices $v^{(j)}$ with $v \in V(H)$ and $j \in (t'-1)i+[t']$.
    By placing non-endpoint vertices from the linkages $Q_1,\ldots,Q_q$ into different layers and keeping all endpoints in the first layer, we obtain a topological $G$-minor model $Q$ in $H \boxtimes K_{qt'}$.
    To complete the proof, we observe that $qt' \leq 2\Delta \cdot t' \leq t$.
\end{proof}

\begin{corollary}
    \label{cor:reroute-G-to-G'}
    Let $H$ be a fixed $k$-vertex graph and let $f(k) = 3 k^{k + 2}$.
    Then there is an $O(n^{f(k)})$ time algorithm that, 
    given an instance $G$ for $\threecol$ with $n$ vertices and maximum degree $4$, 
    either lists all 3-colorings of $G$ or outputs a $\threeass$ instance with graph $G'$ such that
    \begin{enumerate}
        \item $G'$ is the image of a topological $G$-minor in $H \boxtimes K_t$ for $t \coloneqq 12\lceil n / \gamma(H) \rceil$, and 
        \item the proper $3$-colorings of $G$ correspond bijectively to the proper $3$-assignments of $G'$.
    \end{enumerate}
\end{corollary}

\begin{proof}
    Let $n_0$ denote the integer from \Cref{lem:reroute-G-to-G'}.
    If $n < n_0$ then we list all 3-colorings of $G$ by brute-force in constant time.
    Otherwise, we can compute a topological $G$-minor model $Q$ in $H \boxtimes K_t$ in time $O(n^{f(k)})$ by Lemma \ref{lem:reroute-G-to-G'}.
    Let us write $G'$ for the image of $Q$.
    
    We finalize the construction of the $\threeass$ instance by specifying a partition of $E(G')$ into $E_=$ and $E_{\neq}$: For each path $P_{uv}$ in $Q$, place one arbitrary edge into $E_{\neq}$ and all other edges into $E_=$.
    Then the proper $3$-assignments to $G'$ correspond to the proper $3$-colorings of $G$, since contracting all equality edges in $G'$ (which does not change the number of $3$-assignments) yields an isomorphic copy of $G$ on disequality edges.
\end{proof}

Combining the above, the proof of \Cref{thm:linkage-ETH} is complete.

\begin{proof}[Proof of \Cref{thm:linkage-ETH}]
By \Cref{thm:3col-hard}, ETH implies a constant $\beta > 0$ such that no $O(2^{\beta \cdot n})$-time algorithm solves $\threecol$ on $n$-vertex graphs $G$ of maximum degree $4$.
We set $\alpha = \beta / 39$
and derive a contradiction from an $O(s^{\alpha \cdot \gamma})$-time algorithm for $\colsub{H}$ on $s$-vertex input graphs, where $H$ is any fixed graph with $\gamma = \gamma(H)$. 
Moreover, we only need to consider the case where $\gamma \geq 39/\beta$, because otherwise the theorem is trivial.

In the following, let $G$ be an instance for $\threecol$ of maximum degree $4$.
In time $O(n^{f(k)})$, \Cref{cor:reroute-G-to-G'} computes from $G$ an equivalent instance for $\threeass$ with a graph $G' \subseteq H \boxtimes K_t$ for $t = 12\lceil n / \gamma \rceil$ (or lists all 3-colorings of $G$).
In time $9^t \cdot \mathrm{poly}(k,t)$,
\Cref{lem:split-list} then yields a graph $X$ with $|V(X)| \leq k \cdot 3^t$ such that $3$-assignments in $G'$ correspond to colorful $H$-copies in $X$.
The overall running time to construct $X$ is 
\begin{equation}
\label{eq:eth-conclusion-1}
    O(n^{f(k)}) + 9^{12\lceil n / \gamma \rceil} \cdot \mathrm{poly}(k,t) 
\ = \ O(2^{39 n / \gamma})
\ = \ O(2^{\beta \cdot n}).
\end{equation}
In the last step, we use the aforementioned assumption $\gamma \geq 39 / \beta$.
Then use the assumed $O(s^{\beta/39 \cdot \gamma})$-time algorithm for $\colsub{H}$ on $s$-vertex input graphs.
Its running time on the graph $X$ constructed before, with $s \leq k \cdot 3^t$ vertices, is
\begin{equation}
\label{eq:eth-conclusion-2}
O({(k \cdot 3^t)}^{\beta/39 \cdot \gamma}) 
\ = \ O(3^{12\lceil n / \gamma \rceil \cdot \beta/39 \cdot \gamma}) 
\ = \ O(3^{\beta / 2 \cdot n})
\ = \ O(2^{\beta \cdot n}).
\end{equation}

Combining \Cref{eq:eth-conclusion-1,eq:eth-conclusion-2}, we conclude that both constructing $X$ and solving $\colsub{H}$ on $X$ can be achieved in overall time $O(2^{\beta \cdot n})$.
This contradicts \Cref{thm:3col-hard}.
The proof for the counting version is analogous.
\end{proof}

\section{Switching Networks}

\begin{figure}[t]
\centering
\begin{minipage}{0.48\linewidth}
\centering
\begin{tikzpicture}[xscale=1.05,yscale=1.05]
\xdef\magicnum{0.75}
\xdef\rectslack{0.18}
\draw[draw opacity=0, fill opacity=0.1, fill=black!90!white] (1-\rectslack,0*\magicnum-\rectslack) rectangle (4+\rectslack,3*\magicnum+\rectslack);
\draw[draw opacity=0, fill opacity=0.1, fill=black!90!white] (1-\rectslack,4*\magicnum-\rectslack) rectangle (4+\rectslack,7*\magicnum+\rectslack);
\foreach \y in {0,...,7} {
    \pgfmathparse{int(\y+1)} \xdef\yone{\pgfmathresult}
    \node [shape=circle,fill=black,scale=0.5] (V0{\y}) at ({0},{\y*\magicnum}) {};
    \node [shape=circle,fill=black,scale=0.5] (V1{\y}) at ({1},{\y*\magicnum}) {};
    \node [shape=circle,fill=black,scale=0.5] (V4{\y}) at ({4},{\y*\magicnum}) {};
    \node [shape=circle,fill=black,scale=0.5] (V5{\y}) at ({5},{\y*\magicnum}) {};
    \node [scale=0.9] () at ({-0.3},{(7-\y)*\magicnum}) {$v_{\yone}$};
    \node [scale=0.9] () at ({5.3},{(7-\y)*\magicnum}) {$w_{\yone}$};
}
\foreach \y in {1,...,4} {
    \node [scale=0.9] () at (1.3,{5.5*\magicnum+\magicnum*(2.5-\y)*1}) {$v^\uparrow_{\y}$};
    \node [scale=0.9] () at (3.7,{5.5*\magicnum+\magicnum*(2.5-\y)*1}) {$w^\uparrow_{\y}$};
    \node [scale=0.9] () at (1.3,{1.5*\magicnum+\magicnum*(2.5-\y)*1}) {$v^\downarrow_{\y}$};
    \node [scale=0.9] () at (3.7,{1.5*\magicnum+\magicnum*(2.5-\y)*1}) {$w^\downarrow_{\y}$};
}
\foreach \y in {0,...,7} {
    \pgfmathparse{int(\y+4*(1-2*isodd(div(\y,4))))} \xdef\laby{\pgfmathresult}
    \path [] (V0{\y}) edge node [] {} (V1{\laby});
    \path [] (V0{\y}) edge node [] {} (V1{\y});
    \path [] (V4{\y}) edge node [] {} (V5{\laby});
    \path [] (V4{\y}) edge node [] {} (V5{\y});
}
\node [scale=1.2] () at (2.5,1.5*\magicnum) {$B^\downarrow_{2}$};
\node [scale=1.2] () at (2.5,5.5*\magicnum) {$B^\uparrow_{2}$};
\end{tikzpicture}
\end{minipage}
\begin{minipage}{0.48\linewidth}
\centering
\begin{tikzpicture}[xscale=1.05,yscale=1.05]
\xdef\magicnum{0.75}
\xdef\ccnt{3} 
\pgfmathparse{int(\ccnt-1)} \xdef\maxx{\pgfmathresult}
\pgfmathparse{int(pow(2,\ccnt)-1)} \xdef\maxy{\pgfmathresult}

\xdef\rectslack{0.18}
\foreach \x in {\maxx,...,1} {
    \pgfmathparse{int(pow(2,\ccnt-\x)-1)} \xdef\bcntm{\pgfmathresult}
    \pgfmathparse{\ccnt-\x} \xdef\boxlx{\pgfmathresult}
    \pgfmathparse{\ccnt+\x-1} \xdef\boxrx{\pgfmathresult}
    \foreach \t in {0,...,\bcntm} {
        \pgfmathparse{int(pow(2,\x))*\t} \xdef\boxby{\pgfmathresult}
        \pgfmathparse{int(pow(2,\x))*(\t+1)-1} \xdef\boxty{\pgfmathresult}
        \draw[draw opacity=0, fill opacity=0.1, fill=black!90!white] ({\boxlx-\rectslack},{\boxby*\magicnum-\rectslack}) rectangle ({\boxrx+\rectslack},{\boxty*\magicnum+\rectslack});
    }
}

\foreach \x in {0,...,\maxx} {
\foreach \y in {0,...,\maxy} {
    \node [shape=circle,fill=white,scale=0.5] (V{\x}L{\y}) at ({\ccnt-1-\x},{\y*\magicnum}) {};
    \node [shape=circle,fill=white,scale=0.5] (V{\x}R{\y}) at ({\ccnt+\x},{\y*\magicnum}) {};
}
}

\foreach \y in {0,...,\maxy} {
    \pgfmathparse{int(\maxy-\y+1)} \xdef\yback{\pgfmathresult}
    \node [scale=0.9] () at ({-0.3},{(\y)*\magicnum}) {$v_{\yback}$};
    \node [scale=0.9] () at ({1.3+2*\maxx},{(\y)*\magicnum}) {$w_{\yback}$};
    \path [] (V{0}L{\y}) edge node [] {} (V{0}R{\y});
    \pgfmathparse{int(\y+1-2*isodd(\y))} \xdef\laby{\pgfmathresult}
    \path [] (V{0}L{\y}) edge node [] {} (V{0}R{\laby});
}

\foreach \x in {1,...,\maxx} {
\foreach \y in {0,...,\maxy} {
    \pgfmathparse{int(\y+pow(2,\x)*(1-2*isodd(div(\y,pow(2,\x)))))} \xdef\laby{\pgfmathresult}
    \pgfmathparse{int(\x-1)} \xdef\labx{\pgfmathresult}
    \path [] (V{\x}L{\y}) edge node [] {} (V{\labx}L{\laby});
    \path [] (V{\x}L{\y}) edge node [] {} (V{\labx}L{\y});
    \path [] (V{\x}R{\y}) edge node [] {} (V{\labx}R{\laby});
    \path [] (V{\x}R{\y}) edge node [] {} (V{\labx}R{\y});
}
}

\usetikzlibrary {arrows.meta}
\def\benespath#1#2#3#4#5#6#7{
\draw [#7, very thick] (V{2}L{#1}) edge (V{1}L{#2}) (V{1}L{#2}) edge (V{0}L{#3}) (V{0}L{#3}) edge (V{0}R{#4}) (V{0}R{#4}) edge (V{1}R{#5}) (V{1}R{#5}) edge (V{2}R{#6});
\node [shape=circle,fill=#7,scale=0.5] () at (V{2}L{#1}) {};
\node [shape=circle,fill=#7,scale=0.5] () at (V{1}L{#2}) {};
\node [shape=circle,fill=#7,scale=0.5] () at (V{0}L{#3}) {};
\node [shape=circle,fill=#7,scale=0.5] () at (V{0}R{#4}) {};
\node [shape=circle,fill=#7,scale=0.5] () at (V{1}R{#5}) {};
\node [shape=circle,fill=#7,scale=0.5] () at (V{2}R{#6}) {};
}
\benespath{3}{3}{3}{2}{0}{0}{cbfp1}
\benespath{0}{0}{2}{3}{1}{1}{cbfp1}
\benespath{6}{2}{0}{0}{2}{2}{cbfp2}
\benespath{5}{5}{5}{5}{7}{3}{cbfp2}
\benespath{4}{4}{4}{4}{4}{4}{cbfp3}
\benespath{2}{6}{6}{7}{5}{5}{cbfp3}
\benespath{7}{7}{7}{6}{6}{6}{cbfp4}
\benespath{1}{1}{1}{1}{3}{7}{cbfp4}
\path [ultra thick,bend right=45,cbfp1] (V{2}R{0}) edge node [] {} (V{2}R{1});
\path [ultra thick,bend right=45,cbfp2] (V{2}R{2}) edge node [] {} (V{2}R{3});
\path [ultra thick,bend right=45,cbfp3] (V{2}R{4}) edge node [] {} (V{2}R{5});
\path [ultra thick,bend right=45,cbfp4] (V{2}R{6}) edge node [] {} (V{2}R{7});

\end{tikzpicture}
\end{minipage}

\begin{minipage}{0.48\linewidth}
\centering
(a)
\end{minipage}
\begin{minipage}{0.48\linewidth}
\centering
(b)
\end{minipage}
\caption{(a) Recursive construction of Bene\v{s} network $B_3$ with $8$ inputs and $8$ outputs from two copies of $B_2$.
(b) The augmented Bene\v{s} network $\check{B}_3$ is obtained by adding a matching to the outputs of $B_3$, shown as curved edges.
Thick paths indicate an $M$-linkage in $\check{B}_3$ for the matching $M = \{\textcolor{cbfp4}{v_1v_7}, \, \textcolor{cbfp2}{v_2v_3}, \, \textcolor{cbfp3}{v_4v_6}, \, \textcolor{cbfp1}{v_5v_8} \}$ on the input vertices.
}  
\label{fig:benes-network}
\end{figure}

In this section, we consider a construction by Bene\v{s} \cite{Benes64a} that yields $k$-vertex graphs with degree $4$ and a linkage capacity of $\Omega(k / \log k)$.
In particular, this allows us to complete the fully self-contained proof of \cref{thm:CYBT-sparse}.

The Bene\v{s} network $B_\ell$ for $\ell \in \NN$ has $2^\ell$ distinguished input and output vertices.
In our terms, for every matching $M$ between the inputs and outputs, the network $B_\ell$ admits an uncongested $M$-linkage.
By ``short-circuiting'' the outputs, we obtain an \emph{augmented Bene\v{s} network} $\check{B}_\ell$, which allows routing paths from inputs back to inputs.
In our terms, the inputs form a matching-linked set, since every matching $M$ on the inputs admits an uncongested $M$-linkage in $\check{B}_\ell$.

\begin{algorithm}[ht]
\caption{Construct plain Bene\v{s} networks}\label{alg:construct}
\begin{algorithmic}
\Procedure{Benes}{$\ell$} returns $B_\ell$ with $s=2^\ell$ in-/outputs
    \If{$\ell =1$} \Return $K_{2,2}$ with in-/outputs $v_i,w_i$ for $i\in[2]$
    \Else
    \State $B^\uparrow \gets$ \Call{Benes}{$\ell-1$} with in-/outputs $v_i^\uparrowsm,w_i^\uparrowsm$ for $i \in [s/2]$
    \State $B^\downarrow \gets$ \Call{Benes}{$\ell-1$} with in-/outputs $v_i^\downarrowsm,w_i^\downarrowsm$ for $i \in [s/2]$
    \State $B \gets$ vertex-disjoint union of $B^\uparrowsm$ and $B^\downarrowsm$
    \For{$i \in [s/2]$} add to $B$ 
    \State all four edges between $\{v_i,v_{i+s/2}\}$ and $\{v_i^\uparrowsm, v_i^\downarrowsm \}$,
    \State all four edges between $\{w_i,w_{i+s/2}\}$ and $\{w_i^\uparrowsm, w_i^\downarrowsm \}$
    \EndFor
    \State \Return $B$ with in-/outputs $v_i, w_i$ for $i \in [s]$
    \EndIf
\EndProcedure
\end{algorithmic}
\end{algorithm}

\begin{definition}[Bene\v{s} networks]
    The plain \emph{Bene\v{s} network} $B_\ell$ for $\ell \in \NN$ is the graph with distinguished inputs $v_i$ and outputs $w_i$, for $i\in [s]$ with $s=2^\ell$, returned by \textsc{Benes}$(\ell)$ in \Cref{alg:construct}.
    The \emph{augmented Bene\v{s} network} $\check{B}_\ell$ is obtained from $B_\ell$ by adding an edge between outputs $w_{2i-1}$ and $w_{2i}$, for each $i \in [s/2]$.
\end{definition}

A visualization can be found in \Cref{fig:benes-network}.
Both $B_\ell$ and $\check{B}_\ell$ clearly have maximum degree $4$.
Let $T(s)$ for $s = 2^\ell$ count the vertices in the $s$-input Bene\v{s} network $B_\ell$ or $\check{B}_\ell$.
By construction, we have $T(s) = 2 \cdot T(s/2) + 2s$, and thus $T(s) = 2s \log_2 s$.
Bene\v{s} networks are designed to admit uncongested linkages between the inputs and outputs~\cite{Benes64a}:

\begin{theorem} \label{thm:augmented-benes-linkage}
For $\ell \in \NN$, the set $V$ of inputs in $\check{B}_\ell$ is matching-linked, with $|V|=s=2^\ell$.
Moreover, given as input $\ell \in \NN$ and a matching $M$ on $V$, an uncongested $M$-linkage in $\check{B}_\ell$ can be computed in $O(s \log s)$ time.
\end{theorem}

With \Cref{lem:linkage-capacity-via-blowup}, we obtain:

\begin{corollary}
    \label{cor:benes-linkage-capacity}
    For $s = 2^\ell$, we have $\gamma(\check{B}_\ell) \geq s/3$.
\end{corollary}

By combining \cref{thm:linkage-ETH} and \cref{cor:benes-linkage-capacity}, we can give an elementary proof of \cref{thm:CYBT-sparse} (in a slightly modified form; see \Cref{rem:benes-cybt-caveat}).

\begin{theorem}
    \label{thm:CYBT-sparse-simple}
    Assuming ETH, there exists a universal constant $\alpha > 0$ and an infinite sequence of graphs $H_1,H_2,\ldots$ such that, for all $k \in \NN$, the graph $H_k$ has $k$ vertices and maximum degree $4$,
    and $\colsub{H_k}$ does not admit an $O(n^{\alpha \cdot k / \log k})$-time algorithm.
\end{theorem}

\begin{proof}
    For every $k \in \NN$, pick $\ell \in \NN$ maximal such that $|V(\check{B}_\ell)| \leq k$.
    Let $H_k$ be obtained from $\check{B}_\ell$ by adding isolated vertices until the number of vertices is $k$. 
    Since $|V(\check{B}_\ell)| = 2^{\ell+1}\ell$, we conclude that $2^{\ell+1}\ell \leq k < 2^{\ell+2}(\ell+1)$ which implies that $k/\log_2 k < 2^{\ell+2}$.
    So
    \[\gamma(H_k) \geq \gamma(\check{B}_\ell) \geq 2^{\ell}/3 > \frac{1}{12} \cdot k/\log_2 k\]
    by \cref{cor:benes-linkage-capacity}.
    Now the theorem follows from \cref{thm:linkage-ETH}.
\end{proof}

\begin{remark}
\label{rem:benes-cybt-caveat}
    Observe that \cref{thm:CYBT-sparse} provides a sequence of graphs of maximum degree $3$ whereas \cref{thm:CYBT-sparse-simple} ``only'' guarantees maximum degree $4$.
    However, the augmented Bene\v{s} networks $\check{B}_\ell$ can easily be modified to have maximum degree $3$ by replacing each vertex with an edge, so
    \tikz[baseline=-.5*(height("$+$")-depth("$+$")),xscale=0.25,yscale=0.15]{\node [shape=circle,fill=black,scale=0.3] (v) at (0,0) {};
    \node [scale=0.2] (v1) at (1,1) {};
    \node [scale=0.2] (v2) at (1,-1) {};
    \node [scale=0.2] (v3) at (-1,1) {};
    \node [scale=0.2] (v4) at (-1,-1) {};
    \path [draw=black] (v) -- (v1);
    \path [draw=black] (v) -- (v2);
    \path [draw=black] (v) -- (v3);
    \path [draw=black] (v) -- (v4);
    }
    becomes
    \tikz[baseline=-.5*(height("$+$")-depth("$+$")),xscale=0.25,yscale=0.15]{\node [shape=circle,fill=black,scale=0.3] (v) at (0,0) {};
    \node [shape=circle,fill=black,scale=0.3] (v0) at (1,0) {};
    \node [scale=0.2] (v1) at (2,1) {};
    \node [scale=0.2] (v2) at (2,-1) {};
    \node [scale=0.2] (v3) at (-1,1) {};
    \node [scale=0.2] (v4) at (-1,-1) {};
    \path [draw=black] (v0) -- (v1);
    \path [draw=black] (v0) -- (v2);
    \path [draw=black] (v0) -- (v);
    \path [draw=black] (v) -- (v3);
    \path [draw=black] (v) -- (v4);
    },
    and all other relevant properties remain the same.
\end{remark}

For readers familiar with expander graphs, let us also remark that the Bene\v{s} network $B_\ell$ with $s = 2^\ell$ does \emph{not} have constant expansion, as witnessed by its ``upper half'' $U$ that contains the vertices of $B_{\ell-1}^\uparrowsm$ and all inputs and outputs with indices $i \in [s/2]$: 
We have $|U| = s\log_2 s$, but the $2s$ neighbors of $U$ are all contained in the first two and last two columns of $B_\ell$.
This also holds for the augmented $\check{B}_\ell$.

\subsection*{Universality of augmented Bene\v{s} networks}

As an independent point of interest, let us remark that blowups of Bene\v{s} networks  are universal for bounded-degree graphs with respect to topological minor containment:
Mimicking the proof of \Cref{lem:reroute-G-to-G'}, every $n$-vertex graph of maximum degree $\Delta$ can be found as a topological minor in the $2\Delta$-blowup of an augmented Bene\v{s} network with $n$ inputs.

For comparison, every graph that contains every $n$-vertex graph of maximum degree $\Delta$ as a \emph{subgraph} must necessarily have $\Omega(n^{2-2/\Delta})$ edges~\cite{DBLP:conf/focs/AlonCKRRS00}.
Under the more relaxed notion of universality via topological minor containment, Bene\v{s} networks show that universal graphs with only $O(\Delta^2 \cdot n \log n)$ vertices and edges are achievable.

\begin{theorem}
    For every $n,\ell \in \NN$, every graph $G$ of maximum degree $\Delta$ and $n \leq 2^\ell$ vertices is a topological minor of $\check{B}_\ell \boxtimes K_{2\Delta-1}$. Moreover, on input $G$, a topological $G$-minor model in $\check{B}_\ell \boxtimes K_{2\Delta-1}$ can be computed in polynomial time.
\end{theorem}

\begin{proof}
    By \Cref{thm:augmented-benes-linkage}, the inputs in $\check{B}_\ell$ form a matching-linked set $X$ of size $s = 2^\ell \geq n$.
    We view $V(G) \subseteq X$ and decompose $E(G)$ into $2\Delta-1$ matchings via the greedy edge-coloring algorithm. 
    For each matching $M$, we use \cref{thm:augmented-benes-linkage} to find an $M$-linkage $Q$ in $\check{B}_\ell$ and place the internal vertices of $Q$ in a private layer of $\check{B}_\ell \boxtimes K_{2\Delta-1}$, as in the proof of \cref{lem:reroute-G-to-G'}.
    The union of the linkages constructed this way is a topological $G$-minor model in $\check{B}_\ell \boxtimes K_{2\Delta-1}$.
\end{proof}

\section{Patterns of Superlinear Density}

We turn our attention to dense patterns, i.e., $k$-vertex patterns $H$ of average degree $d(H) \in \omega(1)$.
Unlike the sparse setting discussed earlier, a linkage capacity of $\Theta(k)$ is achievable in the dense case, which implies tight lower bounds for $\colsub{H}$ under ETH.

\subsection{Worst Case}
\label{sec:dense-worstcase}

We show that, for every graph $H$, the average degree $d(H) = 2|E(H)|/|V(H)|$ is a lower bound on the linkage capacity of $H$, up to a constant factor. 
First, we use Mader's Theorem \cite[Corollary 1]{Mader72} to extract a highly connected subgraph from $H$.
A graph $H$ is \emph{$\ell$-connected} if $|V(H)| > \ell$ and $H - X$ is connected for every set $X \subseteq V(H)$ with $|X| < \ell$.

\begin{theorem}[see {\cite[Theorem 1.4.3]{graph}}]\label{theo:mader}
    Every graph $H$ with $d(H) \geq 4\ell$ contains a $(\ell+1)$-connected subgraph $H'$ with $d(H') > d(H) - 2\ell$.
\end{theorem}

Second, within the scope of this subsection only, we say that a graph $H$ is \emph{$\ell$-globally linked} if $|V(H)| \geq 2\ell$ and each set $X \subseteq V(H)$ of size at most $2\ell$ is matching-linked in $H$.
(In graph theory, this notion is usually just called $\ell$-\emph{linked}---see, e.g., \cite{graph}. In our paper, we refer to it as \emph{$\ell$-globally linked} to distinguish it from our previous definitions of linkedness.)
This definition implies in particular that $H$ contains a matching-linked set $X$ with $|X| \geq 2\ell$.
Thomas and Wollan \cite[Corollary 1.2]{TW05} show that high connectivity implies high global linkedness. 

\begin{theorem}[see {\cite[Theorem 3.5.3]{graph}}]
    \label{theo:Thomas:Wollan}
    Let $H$ be a graph and $\ell \in \NN$. If $H$ is $2\ell$-connected and $d(H) \geq 16 \ell$, then $H$ is $\ell$-globally linked.
\end{theorem}

Together, these two theorems imply a lower bound on the linkage capacity that is linear in the average degree.
\begin{lemma}
    \label{lem:linkage-capacity-average-degree}
    For every graph $H$, we have $\gamma(H) \geq d(H)/48$.
\end{lemma}
\begin{proof}
    We assume $d(H) \geq 48$ since otherwise the lemma holds trivially because $\gamma(H) \geq 1$ for all graphs $H$.
    \cref{theo:mader} yields a $\lceil d(H)/4 \rceil$-connected subgraph $H'$ of $H$ with $d(H') > d(H)/2$.
    \cref{theo:Thomas:Wollan} shows that $H'$ is $\lceil d(H)/32 \rceil$-globally linked and thus contains a matching-linked set $X$ of size at least $d(H)/16$.
    Then \cref{lem:linkage-capacity-via-blowup} shows that $\gamma(H) \geq \gamma(H') \geq d(H)/48$, where the first inequality uses that $H'$ is subgraph of $H$.
\end{proof}

Now, \cref{thm:dense-hard} follows from \cref{thm:linkage-ETH} and \cref{lem:linkage-capacity-average-degree}.

\subsection{Average Case}

To show the hardness in the average case, we consider the linkage capacity of the Erd\H{o}s-R\'{e}nyi random graph. 
Let $\mathcal G(k,p)$ denote the distribution over $k$-vertex graphs where each edge is included independently with probability $p$. 
We need the following theorem adapted from \cite{BFSU96}, where ``with high probability'' refers to a probability tending to $1$ for $k\to \infty$.

\begin{theorem} \label{thm:BFSU-gnp}
    Let $\varepsilon>0$ be a constant. 
    For all $p\geq (1+\varepsilon)\log (k)/k$ the following holds:
    With high probability, for a random graph $H \sim \mathcal G (k, p)$, every matching $M$ on $V(H)$ can be partitioned into $r = O( \log k / \log kp )$ matchings $M_1, \dots, M_r$ such that $H$ contains an uncongested $M_i$-linkage for all $i\in [r]$.
\end{theorem}

The original theorem statement and proof in \cite[Corollary 1.1]{BFSU96} are concerned with the fixed-sized random graph model $G(k,m)$, and only deals with even $k$. 
But on the other hand, they give a stronger statement concerning the algorithmic efficiency of finding the desired partition, that it can be obtained with high probability by a random partition. 
A proof of the version stated here can be found in the full version.

The last theorem can be used to find large matching-linked sets inside a proper blowup of a random graph, which implies a high linkage capacity by \cref{lem:linkage-capacity-via-blowup}.

\begin{lemma}
    \label{lem:linkage-capacity-random}
    Let $\varepsilon>0$ be a constant. 
    For all $p\geq (1+\varepsilon)\log (k)/k$, the linkage capacity of $H \sim \mathcal G(k,p)$ is at least $\Omega(\frac{k\log(k p)}{\log k})$ with high probability.
\end{lemma}

\begin{proof}
    Let $r$ be the bound specified in \Cref{thm:BFSU-gnp} and consider the blowup graph $H\boxtimes K_{2r}$.
    Let $X \coloneqq \{v^{(1)} \mid v \in V(H)\}$. 
    We show that $X$ is matching-linked in $H\boxtimes K_{2r}$ with high probability, and the lemma then follows using \cref{lem:linkage-capacity-via-blowup}.

    Let $M'$ be a matching on $X$. 
    By the definition of $X$, its $H$-projection, $M \coloneqq \pi(M')$, is also a matching on $V(H)$.  
    We invoke \Cref{thm:BFSU-gnp} on the graph $H$ with respect to the matching $M$ to obtain a partition $M_1,\ldots,M_r$, such that, for all $i\in[r]$ there is an uncongested $M_i$-linkage $Q_i$ in $H$. 
    Then $Q = \bigcup_{i \in [r]} Q_i$ is an $r$-congested $M$-linkage in $H$.
    So there is an uncongested $M'$-linkage in $H\boxtimes K_{2r}$ by Lemma \ref{lem:congestion-vs-blowup}.
\end{proof}

Now, \cref{thm:avg-dense-hard} follows from \cref{thm:linkage-ETH} and \cref{lem:linkage-capacity-random}.

\section{Large-Treewidth Patterns and Concurrent Flows}
\label{sec:treewidth}

In this section, we relate the linkage capacity of a graph to its treewidth.
Towards this end, we first connect the linkage capacity to certain (fractional) multicommodity flows, and afterward rely on existing connections between such flows and treewidth \cite[Section 3.1]{Marx10}.

More specifically, given a graph $H$ and $W \subseteq V(H)$, we consider the following multicommodity flow problem.
For every pair $(u,v) \in W^2$, there is a distinct commodity $\circleuv$ that can be sent in arbitrary fractional amounts along different paths from $u$ to $v$ in $H$.
The goal is to determine whether all pairs $(u,v)$ can concurrently send an $\epsilon$ amount of $\circleuv$ to each other, while the total flow through every vertex $w \in V(H)$ is at most some globally fixed capacity $C$.
Formally, this is captured by the following LP:

\begin{definition}
    Let $H$ be a graph. For $u,v \in V(H)$, write $\paths{H}{u}{v}$ for the set of paths from $u$ to $v$ in $H$; the set $\paths{H}{u}{v}$ for $u = v$ contains only the path $(u)$.
    Given $W \subseteq V(H)$, the concurrent flow LP (for $H$ and $W$) with vertex capacity $C>0$ asks to
    \begin{alignat*}{3}
        &\text{maximize } \varepsilon\\
        &\text{subject to} \quad& \sum_{p \in \paths{H}{u}{v}} x_p &\geq \varepsilon &\quad& \forall u,v \in W\\
                               && \sum_{u,v \in W}\sum_{p \in \paths{H}{u}{v}\colon w \in p} x_p &\leq C && \forall w \in V(G)\\
                               && x_p &\geq 0 && \forall u,v \in W, p \in \paths{H}{u}{v}.
    \end{alignat*}
    We write $\varepsilon(H,W)$ to denote the optimal LP value for capacity $C = 1$.
\end{definition}

While an optimal solution for $C=1$ may assign fractional values to the variables $x_p$, every solution can be scaled to an integral solution, increasing the required capacity and the optimal LP value by the same factor.
This integral solution then induces a congested model of the multigraph $K_{t,(q)}$ in $H$, where $t \coloneqq |W|$ and $q \in \NN$ is suitably chosen, and $K_{t,(q)}$ has $t$ vertices and contains each possible (undirected) edge with multiplicity $q$.

\begin{lemma}
    \label{lem:fractional-flow-to-congested-minor-model}
    Let $H$ be a graph and $W \subseteq V(H)$ be a set of size $t$.
    Then there is some $D \in \NN$ such that $q \coloneqq D \cdot \varepsilon(H,W)$ is an integer and $H$ contains a $D$-congested $K_{t,(q)}$-linkage, where we set $V(K_{t,(q)}) = W$.
\end{lemma}

\begin{proof}
    Let $D$ be the common denominator of the values for all $x_p$ in a (rational) optimal solution of the concurrent flow LP for $H$ and $W$ with capacity $C=1$.
    Scaling all values by $D$ yields an integral solution of value $q \coloneqq D \cdot \varepsilon(H,W)$ for the LP with capacity $D$.
    Now, consider the multiset $Q$ which, for every distinct $u,v \in W$, contains every path $p \in \paths{H}{u}{v}$ with multiplicity $x_p$.
    Then $Q$ is a $D$-congested $K_{t,(q)}$-linkage where $V(K_{t,(q)}) = W$.
\end{proof}

Using this congested $K_{t,(q)}$-linkage, we will establish lower bounds on the linkage capacity of $H$. The following lemma will be useful, as it allows us to route arbitrary multigraphs of bounded degree via short paths in this $K_{t,(q)}$.

\begin{lemma}
    \label{lem:route-multigraph-in-complete-graph}
    Let $q \in \NN$ and let $M$ be a multigraph with $V(M) = [t]$ and maximum degree at most $qt$.
    Then there is an $M$-linkage $Q = (P_{uv})_{uv \in E(M)}$ in $K_t$ such that every edge $e \in E(K_t)$ appears in at most $18q$ paths in $Q$.
\end{lemma}

\begin{proof}
    If $t \leq 12$ the statement trivially holds by choosing $P_{uv} = (u,v)$ for every $uv \in E(M)$.
    So in the remainder of the proof, we assume that $t > 12$.
    
    First observe that $|E(M)| \leq qt^2/2$ since $\deg_M(v) \leq qt$ for all $v \in V(M)$.
    For every $e = uv \in E(M)$ we set $P_{uv} = (u,x_e,v)$ for some suitable \emph{middle vertex} $x_e \in V(K_t) \setminus \{u,v\}$.
    We construct the paths one by one in a greedy fashion.
    Suppose $\mathcal{P}$ is the collection of paths constructed so far.
    We ensure that
    \begin{enumerate}[(a)]
        \item\label{item:route-multigraph-in-complete-graph-1} every edge of $K_t$ appears in at most $18q$ paths in $\mathcal{P}$, and
        \item\label{item:route-multigraph-in-complete-graph-2} every vertex is the middle vertex on at most $qt$ paths in $\mathcal{P}$.
    \end{enumerate}
    Now consider an edge $e = uv \in E(M)$ that is not yet covered by $\mathcal{P}$.
    We argue that there is some $x \in V(K_t) \setminus \{u,v\}$ such that $\mathcal{P} \cup \{(u,x,v)\}$ still satisfies Conditions \ref{item:route-multigraph-in-complete-graph-1} and \ref{item:route-multigraph-in-complete-graph-2}.

    Since $|E(M)| \leq qt^2/2$, there are at most $t/2$ vertices that are the middle vertex of exactly $qt$ paths in $\mathcal{P}$ (i.e., they cannot be selected as a middle vertex).
    Also, the total number of edges incident to $u$ used in $\mathcal{P}$ is at most $3qt$ since $\deg_M(u) \leq qt$ and $u$ is the middle vertex of at most $qt$ paths.
    So there are most $t/6$ vertices $x \in V(G) \setminus \{u\}$ such that the edge $ux$ is \emph{full}, i.e., $ux$ already appears in $18q$ paths in $\mathcal{P}$.
    Similarly, there are most $t/6$ vertices $x \in V(G) \setminus \{v\}$ such that the edge $vx$ is full.
    Since
    \[\frac{t}{2} + 2 \cdot \left(1 + \frac{t}{6}\right) = \frac{5}{6}t + 2 < t,\]
    there exists at least one $x_e \in V(K_t) \setminus \{u,v\}$ such that $\mathcal{P} \cup \{(u,x_e,v)\}$ still satisfies Conditions  \ref{item:route-multigraph-in-complete-graph-1} and \ref{item:route-multigraph-in-complete-graph-2}.
\end{proof}

We can conclude that a large value of the concurrent flow LP implies large linkage capacity.

\begin{theorem}
    \label{thm:linkage-capacity-vs-concurrent-flow}
    Let $H$ be a graph and $W \subseteq V(H)$. Then $\gamma(H) \geq  \varepsilon(H,W) \cdot |W|^2 / 108$.
\end{theorem}

\begin{proof}[Proof of \cref{thm:linkage-capacity-vs-concurrent-flow}]
    Let $D \in \NN$ be the integer obtained from Lemma \ref{lem:fractional-flow-to-congested-minor-model} and set $D' \coloneqq 18 \cdot D$.
    Let $q \coloneqq D \cdot \varepsilon(H,W)$ which, by Lemma \ref{lem:fractional-flow-to-congested-minor-model}, is an integer.
    Observe that $\varepsilon(H,W) \leq 1/|W|$, so $D \geq q \cdot |W|$.
    Finally, let $s \coloneqq q \cdot |W|$.

    Consider the graph $H \boxtimes K_{2D'}$ and let
    \[X \coloneqq \{w^{(i)} \mid w \in W, i \in [s]\}.\]
    We show that $X$ is matching-linked in $H \boxtimes K_{2D'}$.
    Let $M$ be a matching on $X$.
    Let $\widehat{M} \coloneqq \pi(M)$ be the $H$-projection of $M$.
    Observe that $\deg_{\widehat{M}}(w) \leq s = q \cdot |W|$ for all $w \in W$.
    
    Lemma \ref{lem:route-multigraph-in-complete-graph} finds an $\widehat{M}$-linkage $Q = (P_{uv})_{uv \in E(\widehat{M})}$ in $K_{t,(q)}$ (where $V(K_{t,(q)}) = W$) such that every edge of $K_{t,(q)}$ appears in at most $18$ of those paths.
    Moreover, by Lemma \ref{lem:fractional-flow-to-congested-minor-model}, the graph $H$ contains a $D$-congested $K_{t,(q)}$-linkage $Q' = (P_{uv}')_{uv \in E(K_{t,(q)})}$.
    We construct a $D'$-congested $\widehat{M}$-linkage $\widehat{Q} = (\widehat{P}_{uv})_{uv \in E(\widehat{M})}$ in $H$ as follows.
    For every $uv \in E(\widehat{M})$ we obtain $\widehat{P}_{uv}$ from $P_{uv}$ by substituting $P'_e$ for every edge $e$ appearing on $P_{uv}$.
    Clearly, $\widehat{Q}$ is $D'$-congested since $D' = 18 \cdot D$.
    So there is an uncongested $M$-linkage in $H \boxtimes K_{2D'}$ by Lemma \ref{lem:congestion-vs-blowup}. 
    
    Overall, we get that $X$ is matching-linked in $H \boxtimes K_{2D'}$.
    So
    \[\gamma(G) \geq \frac{1}{3} \cdot \frac{|X|}{2 \cdot D'} = \frac{1}{3} \cdot \frac{s \cdot |W|}{36 \cdot D} = \frac{1}{108} \cdot \frac{q \cdot |W|^2}{D} = \frac{1}{108} \cdot \varepsilon(H,W) \cdot |W|^2\]
    by Lemma \ref{lem:linkage-capacity-via-blowup}.
\end{proof}

To bound the linkage capacity by the treewidth, we combine Theorem \ref{thm:linkage-capacity-vs-concurrent-flow} with the following lemma that is (implicitly) shown by Marx \cite{Marx10}.

\begin{lemma}[\cite{Marx10}]
    Let $H$ be a graph of treewidth $t$.
    Then there is a set $W \subseteq V(H)$ such that $|W| = t$ and $\varepsilon(H,W) = \Omega(1/(t \log t))$.
\end{lemma}

The basic idea to prove the lemma is to consider a set $W \subseteq V(H)$ of size $t$ that does not admit balanced separators; large treewidth guarantees such a set (see \cite[Lemma 3.2]{Marx10}). 
Then, using results from \cite{DBLP:journals/siamcomp/FeigeHL08,DBLP:journals/jacm/LeightonR99}, we obtain a bound on the optimal value of the dual LP, which gives $\varepsilon(H,W) = \Omega(1/(t \log t))$ (see \cite[Proof of Lemma 3.6]{Marx10}).

\begin{corollary}
    \label{cor:linkage-capacity-vs-treewidth}
    Let $H$ be a graph of treewidth $t$.
    Then $\gamma(H) = \Omega(t / \log t)$.
\end{corollary}

In particular, combining Theorem \ref{thm:linkage-ETH} and Corollary \ref{cor:linkage-capacity-vs-treewidth} allows us to recover the complexity lower bounds on $\colsub{H}$ proved in \cite{Marx10}.

We note without a proof that the bound in \cref{cor:linkage-capacity-vs-treewidth} is asymptotically optimal, since $k$-vertex 3-regular expander graphs have treewidth $\Theta(k)$ and linkage capacity $\Theta(k/\log k)$ (see also \cite{DBLP:journals/siamdm/AlonM11}).
We complement \cref{cor:linkage-capacity-vs-treewidth} with the following upper bound on the linkage capacity.

\begin{lemma}
    \label{lem:linkage-capacity-vs-treewidth-upper-bound}
    Let $H$ be a graph of treewidth $t$.
    Then $\gamma(H) \leq 3(t+1)$.
\end{lemma}

\begin{proof}
    We first observe that $\tw(H \boxtimes K_q) \leq q(t+1) - 1$ for every $q \in \NN$.
    Indeed, given a tree decomposition $(T,\beta)$ of $H$, we can obtain a tree decomposition for $H \boxtimes K_q$ by adding all copies of $v \in V(H)$ to all bags containing $v$.

    Now, suppose towards a contradiction that $\gamma(H) > 3(t+1)$.
    Then there is some $q \in \NN$ such that $H \boxtimes K_q$ contains a matching-linked set $X \subseteq V(H \boxtimes K_q)$ such that $|X| = 3q(t+1) + 3$.
    By \cite[Lemma 7.20]{CyganFKLMPPS15} there is a \emph{balanced separation for $X$}, i.e., there are sets $A,B \subseteq V(H \boxtimes K_q)$ such that
    \begin{enumerate}
        \item $A \cup B = V(H \boxtimes K_q)$,
        \item $|A \cap B| \leq q(t+1)$,
        \item there is no edge between $A \setminus B$ and $B \setminus A$, and
        \item $|X \cap A| \leq \frac{2}{3}|X|$ and $|X \cap B| \leq \frac{2}{3}|X|$.
    \end{enumerate}
    Hence, $|X \setminus A| \geq \frac{1}{3}|X| = q(t+1) + 1$ and $|X \setminus B| \geq \frac{1}{3}|X| = q(t+1) + 1$.
    Observe that the two sets $X \setminus A$ and $X \setminus B$ are disjoint (since $A \cup B = V(H \boxtimes K_q)$).
    So there is a matching $M$ on the vertex set $X$ (but not necessarily in $H \boxtimes K_q$) containing $q(t+1) + 1$ edges $M' \subseteq M$ with one endpoint in $X \setminus A$ and the other in $X \setminus B$.
    Now consider an uncongested $M$-linkage in $H \boxtimes K_q$ (which exists since $X$ is matching-linked).
    Then every edge $e \in M'$ is realized by a path that needs to visit a vertex from $A \cap B$.
    However, this is a contradiction since $|M'| = q(t+1) + 1$ and $|A \cap B| \leq q(t+1)$.
\end{proof}

The upper bound is also asymptotically optimal by \cref{lem:linkage-capacity-grid}.

\section{Implications for Counting Small Induced Subgraphs}
\label{sec:indsub-main}

We conclude with an application of our lower bounds for the complexity of counting induced $k$-vertex subgraphs.
A \emph{$k$-vertex graph invariant} $\Phi$ is an isomorphism-invariant map from $k$-vertex graphs $H$ to the real numbers.
We consider $\Phi$ to be fixed and wish to sum $\Phi(G[X])$ over all $k$-vertex subsets $X$ of an input graph $G$ to count, e.g., the planar or Hamiltonian induced $k$-vertex subgraphs of $G$.
Formally, for a $k$-vertex graph invariant $\Phi$, the problem $\sindsub{\Phi}$ takes as input a graph $G$, and asks to compute
\[\numindsub{\Phi}{G} \coloneqq \sum_{X \subseteq V(G)} \Phi(G[X]).\]
This problem was first studied in its parameterized version (where $k$ is part of the input) by Jerrum and Meeks \cite{JerrumM15,JerrumM15b,JerrumM17} and received significant attention in recent years \cite{CurticapeanDM17,CurticapeanN25,DorflerRSW22,DoringMW24,DoringMW25,FockeR24,RothS20,RothSW24}.

To determine the complexity of $\sindsub{\Phi}$, recent works usually analyze the \emph{alternating enumerator} to build a generic reduction from $\ccolsub{H}$. 
Formally, the \emph{alternating enumerator} of a graph invariant $\Phi$ on a graph $H$ is defined as\footnote{The precise formula is not relevant here, but we still give it for completeness.}
\[\widehat{\Phi}(H) = (-1)^{|E(H)|} \sum_{S \subseteq E(H)} (-1)^{|S|} \,\Phi(H[S]),\]
where $H[S]$ has vertex set $V(H)$ and edge set $S$.
Now, suppose $H$ is a $k$-vertex graph with $\widehat{\Phi}(H) \neq 0$.
Then the problem $\ccolsub{H}$ can be reduced to $\sindsub{\Phi}$ in polynomial time (see, e.g., \cite{CurticapeanN25}).
Hence, we also obtain new lower bounds for $\sindsub{\Phi}$ assuming $\widehat{\Phi}(H) \neq 0$ for suitable graphs $H$.
In particular, we obtain the following result via Theorem~\ref{thm:dense-hard}, which improves over the corresponding lower bound in \cite[Theorem 3.2(a)]{CurticapeanN25} (see also \cite{DorflerRSW22} and \cite[Lemma 2.2]{DoringMW24}).

\begin{theorem}
    \label{thm:indsub-hardness-eth}
    There is a universal constant $\alpha_{\textsc{ind}} > 0$ and an integer $N_0 \geq 1$ such that for all numbers $k,\ell \geq 1$, the following holds:
    If $\Phi$ is a $k$-vertex graph invariant and there exists a graph $H$ with $\widehat{\Phi}(H) \neq 0$ and $E(H) \geq k \cdot \ell \geq N_0$, then $\sindsub{\Phi}$ cannot be solved in time $O(n^{\alpha_{\textsc{ind}} \cdot \ell})$ unless ETH fails.
\end{theorem}

The proof of \cref{thm:indsub-hardness-eth} is given in Appendix \ref{sec:app-indsub}.
As pointed out in \Cref{sec:intro-applications}, the weaker version, which only rules out an exponent of $\alpha_{\textsc{ind}} \cdot \ell/\sqrt{\log \ell}$, has been used to derive various lower bounds for specific types of invariants in \cite{CurticapeanN25,DoringMW24,RothSW24}.
All these lower bounds are improved by our new results.
Let us give one concrete example, which improves over \cite[Corollary 5.1]{CurticapeanN25}.
For a $k$-vertex graph invariant $\Phi$, we write $\supp(\Phi)$ for the set of all graphs $H$ with $V(H) = [k]$ and $\Phi(H) \neq 0$.

\begin{corollary}
    \label{cor:indsub-hardness-small}
    For every $0 < \varepsilon < 1$ there are $N_0,\delta > 0$ such that the following holds.
    Let $k \geq N_0$ and let $\Phi$ be a $k$-vertex graph invariant with $1 \leq |\supp(\Phi)| \leq (2 - \varepsilon)^{\binom{k}{2}}$.
    Then no algorithm solves $\sindsub{\Phi}$ in time $O(n^{\delta \cdot k})$ unless ETH fails.
\end{corollary}

We stress that the exponent in the lower bound of \cref{cor:indsub-hardness-small} is asymptotically optimal.

\begin{proof}
    Let $0 < \varepsilon < 1$.
    By \cite[Theorem 5.1]{CurticapeanN25} there is some $\delta' > 0$ such that for every $k \geq 1$ and every $k$-vertex graph invariant $\Phi$ satisfying the condition of the theorem,
    there is a $k$-vertex graph $H$ such that $\widehat{\Phi}(H) \neq 0$ and $|E(H)| \geq \delta' \cdot \binom{k}{2}$.
    Therefore, the theorem follows from \cref{thm:indsub-hardness-eth} by setting $\delta \coloneqq \frac{1}{3} \cdot \delta' \cdot \alpha_{\textsc{ind}}$ and choosing $\ell \coloneqq \frac{1}{3} \cdot \delta' \cdot k$.
\end{proof}

For the other implications of \cref{thm:indsub-hardness-eth}, we refer the reader to \cite{CurticapeanN24} (the latest arXiv version of \cite{CurticapeanN25}).
Indeed, following the first publication of this work, the latest version of \cite{CurticapeanN25} contains updated lower bounds for $\sindsub{\Phi}$ based on \cref{thm:indsub-hardness-eth}.
However, let us stress that these improved lower bounds should (at least in part) be attributed to this work.

\bibliographystyle{plainurl}
\bibliography{refs}

\clearpage

\appendix

\section{Uncolored vs.\ Colored Subgraph Isomorphism}
\label{app:uncolored}

In this section, we discuss the insights necessary to obtain \Cref{cor:uncolored-sparse,cor:uncolored-avg-dense-hard}.
All details can also be found in \cite{Marx10}.

A homomorphism from a graph $H_1$ to a graph $H_2$ is a mapping $\varphi\colon V(H_1) \to V(H_2)$ such that $\varphi(v)\varphi(w) \in E(H_2)$ for every edge $vw \in E(H_1)$.
Recall that a graph $H$ is a \emph{homomorphic core} if every homomorphism from $H$ to itself is injective.
Observe that an injective homomorphism from a finite graph $H$ to itself necessarily is an automorphism of $H$.

Now, let $H$ be a fixed graph with $k$ vertices.
Also, let $G = (V,E,c)$ be a colored graph where $c\colon V(G) \to V(H)$.
We define the graph  $G_{\upharpoonright H} = (V,E_{\upharpoonright H},c)$ where
\[E_{\upharpoonright H} \coloneqq \{vw \in E \mid c(v)c(w) \in E(H)\}.\]
Also, we write $G_{\upharpoonright H}^\circ = (V,E_{\upharpoonright H})$ to denote the uncolored version of $G_{\upharpoonright H}$.

\begin{lemma}
    \label{lem:subgraph-count-core}
    We have $\numsub{\can{H}}{G} = \numsub{\can{H}}{G_{\upharpoonright H}}$.
    Additionally, if $H$ is a homomorphic core, then $\numsub{\can{H}}{G_{\upharpoonright H}} = \numsub{H}{G_{\upharpoonright H}^{\circ}}$.
\end{lemma}

\begin{proof}
    The first part is trivial, since the edges from $E(G) \setminus E(G_{\upharpoonright H})$ cannot be used to form a subgraph isomorphic to $\can{H}$.

    For the second part, let $G' = (V',E')$ be a subgraph of $G_{\upharpoonright H}^{\circ}$ that is isomorphic to $H$.
    Let $\varphi\colon V(H) \to V'$ be an isomorphism from $H$ to $G'$.
    Then $\psi\colon V(H) \to V(H)$ defined via $\psi(v) = c(\varphi(v))$ is a homomorphism from $H$ to itself.
    Since $H$ is a core, it follows that $\psi$ is an automorphism of $H$.
    This implies that the colored version $G' = (V',E',c_{\upharpoonright V'})$ is isomorphic to $\can{H}$.
    In particular, $\numsub{\can{H}}{G_{\upharpoonright H}} = \numsub{H}{G_{\upharpoonright H}^{\circ}}$.
\end{proof}

In particular, if $H$ is a homomorphic core, we obtain a linear-time reduction from $\colsub{H}$ to $\sub{H}$:
On input $G$, we compute $G_{\upharpoonright H}^{\circ}$ in linear time and use it as input for an algorithm for $\sub{H}$.
This provides the missing piece to obtain \Cref{cor:uncolored-avg-dense-hard} from \Cref{thm:avg-dense-hard} and \cite{BonatoP09}.

Next, we turn to the proof of \Cref{cor:uncolored-sparse} which is based on the following lemma.

\begin{lemma}[see {\cite[Lemma 6.5]{Marx10}}]
    \label{lem:core}
    Given a graph $H$ without isolated vertices, there exists a graph $H'$ such that 
    \begin{itemize}
        \item $H'$ is a homomorphic core,
        \item $H$ is a topological minor of $H'$,
        \item $|V(H')| \leq 7(|E(H)| + |V(H)| + 1)$ and $|E(H')| \leq 13(|E(H)| + |V(H)| + 1)$, and
        \item the maximum degree of $H'$ is $\max(5, \Delta +1)$, where $\Delta$ is the maximum degree of $H$.
    \end{itemize}
\end{lemma}

The proof of this lemma is identical to \cite{Marx10}, but we still include it for completeness.

\begin{proof}
    We start by constructing an auxiliary graph $T_n$ with vertices $v_i$, $i \in [n]$, and $x_i,y_i$, $i \in \{2,\dots,n\}$.
    For every $i \in \{2,\dots,n\}$, we add the edges $x_iy_i$, $v_{i-1}x_i$, $v_{i-1}y_i$, $v_ix_i$ and $v_iy_i$.
    We also add the edges $v_1v_n$ and $v_nx_{n-1}$ to $T_n$.
    Observe that $|V(T_n)| \leq 3n$, $|E(T_n)| \leq 5n$ and the maximum degree of $T_n$ is 4.

    \begin{claim}
        \label{claim:tn-core}
        $T_n$ is a homomorphic core.
        Additionally, for every homomorphism $\varphi$ from $T_n$ to $T_n$, it holds that $\varphi(v_i) = v_i$ for all $i \in [n]$.
    \end{claim}
    
    \begin{claimproof}
        First observe that $T_n$ is not 3-colorable.
        Indeed, in a $3$-coloring the vertices $v_{i-1}$ and $v_i$ need to get the same color for every $i \in {2,\dots,n}$, since both vertices form a triangle together with $x_i$ and $y_i$.
        However, this is impossible since $v_1v_n$ is an edge of $T_n$.
        Additionally, it is easy to see that $T_n$ becomes 3-colorable after deleting any vertex.
        This immediately implies that $T_n$ is a homomorphic core since the chromatic number of the image of a homomorphism $\varphi$ is always at least as the chromatic number of the original graph.
        Thus, $\varphi$ has to be injective and therefore an automorphism.
        
        Next, we argue that $\varphi(v_i) = v_i$ for all $i \in [n]$.
        First note that $v_2,\dots,v_n,x_{n-1}$ are the only vertices of degree 4.
        Moreover, $v_1$ is the only vertex of degree 3 that is only part of a single triangle (each $x_i/y_i$ is part of two triangles).
        Thus $\varphi(v_1) = v_1$.
        Lastly, note that $\varphi$ perseveres distances and  each $v_i$, $i \in \{2,\dots,n\}$, is uniquely described (among $v_2,\dots,v_n,x_{n-1}$) by its distance to $v_1$.
        So we get $\varphi(v_i) = v_i$ for all $i \in [n]$. 
    \end{claimproof} 

    Now, we turn to the construction of $H'$.
    We start by subdividing every edge of $H$ to obtain a new graph $B$ with $|V(B)| = |V(H)| + |E(H)|$, $|E(B)| = 2 \cdot |E(H)|$, and maximum degree $\Delta$.
    Observe that $B$ is a bipartite graph.
    Let $n \coloneqq |V(B)|$ and let $w_1, \dots, w_n$ denote the vertices of $B$.
    We define the graph $H'$ by adding a copy of the graph $T_{2n + 1}$ to $B$, and adding an edge between $w_i$ and $v_{2i}$ for every $i \in [n]$.
    Observe that $H$ is a topological minor if $H'$ since $H$ is a topological minor of $B$ and $B$ is a subgraph of $H'$.
    Also, the maximum degree of $H$ is $\max(5, \Delta +1)$ since the maximum degree of $T_{2n+1}$ is $4$ and each vertex of $T_{2n+1}$ is connected to at most one vertex in $B$.
    Next, observe that
    \begin{align*}
        |V(H')| &= |V(B)| + |V(T_{2n+1})| \leq n + 3(2n+1) = 7(n + 1) \\
        |E(H')| &= |E(B)| + |E(T_{2n+1})| + n \leq 2|E(H)| + 5(2n + 1) + n \leq 13(n + 1)
    \end{align*}
    Lastly, we verify that $H'$ is a core.
    Let $\varphi$ be a homomorphism from $H'$ to $H'$.
    First observe that the vertices of $H'$, which appear in a triangle, are exactly the vertices of $T_{2n+1}$ since $B$ is bipartite.
    Since $\varphi$ maps triangles to triangles, we get that $\varphi$ maps vertices of $T_{2n+1}$ to vertices of $T_{2n+1}$.
    It follows that the restriction of $\varphi$ to $T_{n+1}$ is a homomorphism from $T_{2n + 1}$ to $T_{2n+1}$. 
    Hence, $\varphi$ injective on $T_{2n + 1}$ and $\varphi(v_i) = v_i$ for all $i \in [2n+1]$ by Claim \ref{claim:tn-core}.
    
    Now, we argue that $\varphi(w_i) = w_i$ for all $i \in [n]$ which yields that $\varphi$ is injective on the whole graph $H'$. 
    Let $w_j$ be a neighbor of $w_i$ (recall that $H$ has no isolated vertices).
    Now, the path $P = (v_{2i},w_i,w_j,v_{2j})$ has length 3. Since $\varphi$ maps $v_{2i}$ to $v_{2i}$, $v_{2j}$ to $v_{2j}$, and $v_{2i}$ is non-adjacent to $v_{2j}$, we obtain that $\varphi$ maps $P$ to a path $\varphi(P)$ of length $3$ from $v_{2i}$ to $v_{2j}$.
    However, since the distance of $v_{2i}$ and $v_{2j}$ in $T_{2n + 1}$ is greater than $3$, the path $\varphi(P)$ leaves $T_{2n + 1}$.
    Since $w_i$ is the only vertex outside of $T_{2n + 1}$ that is adjacent to $v_{2i}$, we get that $\varphi(w_i) = w_i$. 
\end{proof}

\begin{proof}[Proof of \Cref{cor:uncolored-sparse}]
    For every $k \in \NN$, pick $\ell \in \NN$ maximal such that $|V(\check{B}_\ell)| \leq k$, and set $H_k \coloneqq \check{B}_\ell$.
    Since $|V(\check{B}_\ell)| = 2^{\ell+1}\ell$, we conclude that $2^{\ell+1}\ell \leq k < 2^{\ell+2}(\ell+1)$ which implies that $k/\log_2 k < 2^{\ell+2}$.
    So
    \[\gamma(H_k) = \gamma(\check{B}_\ell) \geq 2^{\ell}/3 > \frac{1}{12} \cdot k/\log_2 k\]
    by \cref{cor:benes-linkage-capacity}.

    Let $H_k'$ be the graph obtained from applying \Cref{lem:core} to $H_k$.
    For every $k \in \NN$ we get that $H_k'$ is a homomorphic core, $V(H_k') \leq 24 k$, and $H_k'$ has maximum degree 5.
    Also, $\gamma(H_k') \leq \gamma(H_k)$ using \Cref{lem:minor-closed} and the fact that $H_k$ is a topological minor of $H'_k$.

    By \Cref{thm:linkage-ETH}, there exists a universal constant $\alpha > 0$ such that for all $k \in \NN$, no algorithm solves $\colsub{H_k'}$ in time $O(n^{\alpha \cdot k / \log k})$, unless ETH fails.
    Since every graph $H_k'$ is a homomorphic core, we conclude that for all $k \in \NN$, no algorithm solves $\sub{H_k'}$ in time $O(n^{\alpha \cdot k / \log k})$, unless ETH fails (see \Cref{lem:subgraph-count-core} and the linear-time reduction discussed below the lemma).
\end{proof} 

\section{Computing Linkages}\label{app:linkage-capacity}

In this section we prove \Cref{theo:find-linked-set}. Towards this, we first consider the following claim:
\begin{claim}\label{claim:algo:linkage}
    There exists an algorithm $\mathbb{A}$ that, given a set $X \subseteq V(H \boxtimes K_t)$ and a matching $M$ on vertex set $X$, checks if there is an uncongested $M$-linkage in $H \boxtimes K_t$ in time $O(t^{f(k)})$.
\end{claim}
\begin{claimproof}
    Our algorithm works in two steps. In the first step, we select a collection of paths in $H$. In the second step, we embed the paths into $H \boxtimes K_t$.
    First, for each pair of vertices $u, v \in H$, we write $\mathcal{P}_{u, v}$ for the set of simple paths from $u$ to $v$ in $H$. We write $\mathcal{P}$ for the union of all $\mathcal{P}_{u, v}$.
    Note that $|\mathcal{P}| \leq {k^{2 + k}}$ and that we can enumerate all elements in $\mathcal{P}$ in time $O(k^{3 + k})$. 

    In the next step, we consider all multisets $Q$ over $\mathcal{P}$ with multiplicity at most $t$.\footnote{This means that each path $P$ appears at most $t$ times in $Q$.} We check if $Q$ defines an uncongested $M$-linkage in $H \boxtimes K_t$ using the following subroutine.
    
    Let $V \coloneqq V(H \boxtimes K_t) \setminus V(M)$, where $V(M)$ denotes the endpoints of all edges in $M$, and $N \coloneq M$.
    Also set $Q' \coloneqq \emptyset$.
    For an edge $v^{(i)}u^{(j)} \in N$, check if there is a path $P = (v, a_2, \dots, a_{s-1}, u)$ in $Q$.
    If so, check if there is a vertex $a_i^{(b_i)}$ in $V$ for each $i \in \{2,\dots,s-1\}$.
    This defines a path $P' = (v^{(i)}, a_1^{(b_1)}, \dots, a_{s-1}^{(b_{s-1})}, u^{j})$ in $H \boxtimes K_t$ that goes from $v^{(i)}$ to $u^{(j)}$ since $P$ is a path in $H$.
    We add $P'$ to $Q'$.
    Lastly, remove $P$ from $Q$, $v^{(i)}u^{(j)}$ from $N$, and the vertices of $P'$ from $V$.
    We repeat this process until either we abort, or $N$ is empty.
    If $N$ is empty, then $Q'$ defines a valid uncongested $M$-linkage since for each edge $v^{(i)}u^{(j)} \in M$ there is a path in $Q'$ from $v^{(i)}$ to $u^{(j)}$, and each vertex appears in at most one path. This subroutine runs in time $O(t \cdot k^3)$ since there are at most $t \cdot k$ many edges in $M$, we can find a path $P$ in $Q$ in time $O(\log(k^{2+k}))$ using some basic data structure, and we can embed $P$ in $H \boxtimes K_t$ in time $O(k)$.

    We apply this subroutine to all collections $Q$ of $\mathcal{P}$ with multiplicity at most $t$.  If we find an uncongested $M$-linkage this way, we return that $X$ contains an uncongested $M$-linkage. Otherwise, we return false. Since there are at most $(t+1)^{k^{2 + k}}$ many possible collections, the algorithm runs in $O((t+1)^{k^{2 + k}} \cdot t k^3)$ which is in $O(t^{f(k)})$ for $f(k) = 3k^{2 + k}$.
\end{claimproof}

This allows us to prove the main result of this section.

\begin{proof}[Proof of \Cref{theo:find-linked-set}]
    To find a matching-linked set $X$ of maximum size, we first observe that each vertex set $X \subseteq H \boxtimes K_t$ corresponds to a multiset $X'$ of $H$ where $X'$ contains the vertex $v$ with multiplicity $|\{v^{(i)} \in X \mid i \in [t]\}|$.
    Thus $X'$ has multiplicity at most $t$.

    It is easy to see that if two vertex sets $X$ and $Y$ define the same multiset $X'$ then there is an automorphism in $H \boxtimes K_t$ that maps a vertex $v^{(i)}$ only to a vertex of the form $v^{(j)}$ and that maps $X$ to $Y$. 
    Given a matching $M$ on $X$ and a matching $N$ on $Y$, this automorphism also maps each uncongested $M$-linkage to an uncongested $N$-linkage.
    Thus $X$ is a matching-linked set if and only if $Y$ is a matching-linked set. So, we only have to consider the $(t + 1)^{k}$ many multisets of $H$ with multiplicity at most $t$.

    Further, we apply the same idea to matchings $M$ in $H \boxtimes K_t$. For each edge $v^{(i)}u^{(j)} \in M$, we add $vu$ to a multiset $M'$. Observe that $M'$ has multiplicity at most $t$. It is now easy to see that if two matchings $M$ and $N$ correspond to the same multiset $M'$ then $H \boxtimes K_t$ contains a $M$-linkage if and only if $H \boxtimes K_t$ contains a $N$-linkage. Thus we only have to consider the $(t + 1)^{k^2}$ many multisets $M'$ with multiplicity at most $t$.

    Lastly, to find the maximal linked set we first iterate over all multisets $X'$ of $V(H)$ with multiplicity at most $t$. Each $X'$ defines a vertex set $X = \{v^{(i)} : 1 \leq i \leq d(X'_v)\}$ where $d(X'_v)$ is the multiplicity of the vertex $v$ in $X'$. To check if $X$ is matching-linked, we iterate over all multisets $M'$ with multiplicity at most $t$. Each set $M'$ defines a matching $M$ in the following way. Let $V \coloneqq V(H \boxtimes K_t)$. For each $vu \in M'$, we add $v^{(i)}u^{(j)}$ to $M$ where $i \coloneqq \min(s : v^{(s)} \in V)$ and  $j \coloneqq \min(s : u^{(s)} \in V)$.\footnote{If $v = u$, then we add $v^{(i)} v^{(i+1)}$ to $M$.} If not possible, we just ignore $M'$ and consider the next multiset. Then remove  $v^{(i)}$ and $u^{(j)}$ from $V$ and continue with the next element in the multiset $M'$.

    This way, we obtain a matching $M$ that corresponds to $M'$. We check $M$ is a matching on $X$. If not, we check the next multiset $M'$. Otherwise, we use algorithm $\mathbb{A}$ from \cref{claim:algo:linkage} to check if $X$ is a $M$-linkage. If this is not the case, then $X$ is not matching-linked. After considering all $(t+1)^{k^2}$ many multiset $M'$, we know if $X$ is matching-linked. By applying this procedure to all multisets $X'$, we find the largest matching-linked set $X$ in time 
    \[O( \underbrace{(t+1)^k}_{\text{vertex sets}} \cdot \underbrace{(t + 1)^{k^2}}_{\text{matchings}} \cdot  \underbrace{(t+1)^{k^{2 + k}} \cdot t  k^3}_{\text{Algorithm } \mathbb{A}}),\]
    which is in $O(t^{f(k)})$ for $f(k) = 3k^{2 + k}$.
\end{proof}

\section{Linkages in Bene\v{s} Networks}
\label{app:benes}

We prove \Cref{thm:augmented-benes-linkage}.
First, \Cref{alg:route} constructs linkages between inputs and outputs in plain Bene\v{s} networks, as shown in \Cref{lem:benes-route-correct}.
Such linkages in plain Bene\v{s} networks then readily imply linkages among inputs in augmented Bene\v{s} networks.

\begin{figure}
\begin{minipage}{0.65\linewidth}
\begin{algorithm}[H]
\caption{Routing in plain Bene\v{s} networks}\label{alg:route}
\begin{algorithmic}[1]
\Procedure{BenesLink}{$\ell$, $M$} returns an $M$-linkage in $B_\ell$
    \If{$\ell =1$} \Return $M$ (interpreted as linkage) \EndIf
        \State $s \gets 2^\ell$
        \State Define $\pi\colon[s] \to [s]$ via $\pi(i) = j$ for $v_iw_j \in M$
        \State Define $t(i) \coloneqq
            \begin{cases}
                i,&i\leq s/2\\
                i-s/2,&i>s/2
            \end{cases}$
        \State Let $D = ([s],\{\textcolor{blue}{\{i,i+s/2\}} \mid i \in [s/2]\}$
        \Statex $\quad \cup \quad \{\textcolor{red}{\{\pi^{-1}(j),\pi^{-1}(j+s/2)\}} \mid j \in [s/2]\})$\label{line:conflict}
        \State Compute a proper $2$-coloring $f\colon [s] \to \{\uparrow,\downarrow\}$ of $D$ \label{line:conflict-bipartite}
        \Statex \Comment{the path starting in $v_i$ is routed via $B^\uparrowsm_{\ell-1}$ or $B^\downarrowsm_{\ell-1}$}
        \State $M^{\uparrowsm}\gets\left\{v_{t(i)}^{\uparrowsm}w_{t(j)}^{\uparrowsm} \mid f(i) =~\uparrow, v_iw_j\in M\right\}$
    	\State $\{P_{i} \mid f(i) =~\uparrow\} \gets$ \Call{BenesLink}{$\ell-1$, $M^\uparrowsm$} \label{line:B-up-linkage} 
        \State $M^{\downarrowsm}\gets\left\{v_{t(i)}^{\downarrowsm}w_{t(j)}^{\downarrowsm} \mid f(i) =~\downarrow, v_iw_j\in M\right\}$
        \State $\{P_{i} \mid f(i) =~\downarrow\} \gets$ \Call{BenesLink}{$\ell-1$, $M^\downarrowsm$} \label{line:B-down-linkage} 
        \Statex\Comment{find linkages in $B^\uparrowsm_{\ell-1}$ and $B^\downarrowsm_{\ell-1}$}
        \For{$i \in [s]$}
            \State Add path $(v_i,v_{t(i)}^{f(i)}) \circ P_i \circ (w_{t(j)}^{f(i)},w_j)$ to $Q$
            \Statex where $j \coloneqq \pi(i)$\label{line:build-path}
        \EndFor
        \State \Return $Q$
\EndProcedure
\end{algorithmic}
\end{algorithm}
\end{minipage}
\quad
\begin{minipage}{0.3\linewidth}
\centering

\begin{tikzpicture}[xscale=0.9,yscale=0.9]
\hspace{0.1cm}
\node [shape=circle,fill=black,scale=0.5,fill opacity=0.3] (FF1) at ( 112.5:1.6) {};
\node [shape=circle,fill=black,scale=0.5,fill opacity=0.3] (FF2) at (  22.5:1.6) {};
\node [shape=circle,fill=black,scale=0.5,fill opacity=0.3] (FF3) at ( -67.5:1.6) {};
\node [shape=circle,fill=black,scale=0.5,fill opacity=0.3] (FF4) at (-157.5:1.6) {};
\node [shape=circle,fill=black,scale=0.5,fill opacity=0.3] (FF5) at (  67.5:1.6) {};
\node [shape=circle,fill=black,scale=0.5,fill opacity=0.3] (FF6) at ( -22.5:1.6) {};
\node [shape=circle,fill=black,scale=0.5,fill opacity=0.3] (FF7) at (-112.5:1.6) {};
\node [shape=circle,fill=black,scale=0.5,fill opacity=0.3] (FF8) at ( 157.5:1.6) {};

\node [] () at ( 112.5:2) {1};
\node [] () at (  22.5:2) {2};
\node [] () at ( -67.5:2) {3};
\node [] () at (-157.5:2) {4};
\node [] () at (  67.5:2) {5};
\node [] () at ( -22.5:2) {6};
\node [] () at (-112.5:2) {7};
\node [] () at ( 157.5:2) {8};

\path [very thick, draw=blue] (FF1) to [bend left] (FF5);
\path [very thick, draw=blue] (FF2) to [bend left] (FF6);
\path [very thick, draw=blue] (FF3) to [bend left] (FF7);
\path [very thick, draw=blue] (FF4) to [bend left] (FF8);

\path [very thick, draw=red] (FF2) to [] (FF1);
\path [very thick, draw=red] (FF5) to [] (FF4);
\path [very thick, draw=red] (FF7) to [bend left] (FF3);
\path [very thick, draw=red] (FF8) to [] (FF6);
    
\end{tikzpicture}

\vspace{10pt}

\begin{tikzpicture}[xscale=0.95,yscale=1.05]
\xdef\magicnum{0.75}
\xdef\rectslack{0.18}
\draw[draw opacity=0, fill opacity=0.06, fill=black!70!white] (1-\rectslack,0*\magicnum-2*\rectslack) rectangle (3+\rectslack,3*\magicnum+1.5*\rectslack);
\draw[draw opacity=0, fill opacity=0.06, fill=black!70!white] (1-\rectslack,4*\magicnum-2*\rectslack) rectangle (3+\rectslack,7*\magicnum+1.5*\rectslack);
\foreach \y in {0,...,7} {
    \pgfmathparse{int(\y+1)} \xdef\yone{\pgfmathresult}
    \node [shape=circle,fill=black,scale=0.5,fill opacity=0.3] (V0{\y}) at ({0},{\y*\magicnum}) {};
    \node [shape=circle,fill=black,scale=0.5,fill opacity=0.3] (V1{\y}) at ({1},{\y*\magicnum}) {};
    \node [shape=circle,fill=black,scale=0.5,fill opacity=0.3] (V4{\y}) at ({3},{\y*\magicnum}) {};
    \node [shape=circle,fill=black,scale=0.5,fill opacity=0.3] (V5{\y}) at ({4},{\y*\magicnum}) {};
    \node [scale=0.9] () at ({-0.3},{(7-\y)*\magicnum}) {$v_{\yone}$};
    \node [scale=0.9] () at ({4.3},{(7-\y)*\magicnum}) {$w_{\yone}$};
}
\foreach \y in {0,...,7} {
    \pgfmathparse{int(\y+4*(1-2*isodd(div(\y,4))))} \xdef\laby{\pgfmathresult}
    \path [opacity=0.1] (V0{\y}) edge node [] {} (V1{\laby});
    \path [opacity=0.1] (V0{\y}) edge node [] {} (V1{\y});
    \path [opacity=0.1] (V4{\y}) edge node [] {} (V5{\laby});
    \path [opacity=0.1] (V4{\y}) edge node [] {} (V5{\y});
}

\def\benespath#1#2#3#4#5#6#7{
\draw [#7, very thick] (V0{#1}) edge (V1{#2});
\draw [#7, thick, dashed] (V1{#2}) edge (V4{#5});
\draw [#7, very thick] (V4{#5}) edge (V5{#6});
\node [shape=circle,fill=#7,scale=0.5] () at (V0{#1}) {};
\node [shape=circle,fill=#7,scale=0.5] () at (V1{#2}) {};
\node [shape=circle,fill=#7,scale=0.5] () at (V4{#5}) {};
\node [shape=circle,fill=#7,scale=0.5] () at (V5{#6}) {};
}
\benespath{3}{3}{3}{2}{0}{0}{black}
\benespath{0}{0}{2}{3}{1}{1}{black}
\benespath{6}{2}{0}{0}{2}{2}{black}
\benespath{5}{5}{5}{5}{7}{3}{black}
\benespath{4}{4}{4}{4}{4}{4}{black}
\benespath{2}{6}{6}{7}{5}{5}{black}
\benespath{7}{7}{7}{6}{6}{6}{black}
\benespath{1}{1}{1}{1}{3}{7}{black}

\foreach \y in {1,...,4} {
    \node [scale=0.9] () at (1.3,{5.5*\magicnum+\magicnum*(2.5-\y)*1}) {$v^\uparrow_{\y}$};
    \node [scale=0.9] () at (2.7,{5.5*\magicnum+\magicnum*(2.5-\y)*1}) {$w^\uparrow_{\y}$};
    \node [scale=0.9] () at (1.3,{1.5*\magicnum+\magicnum*(2.5-\y)*1}) {$v^\downarrow_{\y}$};
    \node [scale=0.9] () at (2.7,{1.5*\magicnum+\magicnum*(2.5-\y)*1}) {$w^\downarrow_{\y}$};
}

\end{tikzpicture}

\vspace{-10pt}
\end{minipage}
\vspace{0.2cm}

\caption{
The algorithm for routing in the Bene\v{s} network (left), the conflict graph $D$ (top right), and the edges chosen on the top level of recursion (bottom right) based on the example in \Cref{fig:benes-network}. 
Here, $s=2^3=8$, and $\pi=(\pi(1),\dots,\pi(8))=(2,6,5,4,8,3,1,7)$. 
Because $v_1$ and $v_5$ cannot both occupy $v_1^{\uparrowsm}$ (or $v_1^{\downarrowsm}$) at the same time, there is an edge $\textcolor{blue}{\{1,5\}}$ in $D$. 
Because $w_4$ and $w_8$ cannot both occupy $w_4^{\uparrowsm}$ (or $w_4^{\downarrowsm}$) at the same time, there is an edge between the preimages of $w_4$ and $w_8$ under $\pi$, which is $\{\pi^{-1}(4),\pi^{-1}(8)\}=\textcolor{red}{\{4,5\}}$. 
The resulting 2-coloring of $D$ is $f=(f(1),\dots,f(8))=(\uparrow,\downarrow,\uparrow,\uparrow,\downarrow,\uparrow,\downarrow,\downarrow)$, giving the routing at the current recursive level as shown on the bottom right.
Dashes edges correspond to the matchings $M^{\uparrowsm}$ and $M^{\downarrowsm}$ for the next level of recursive calls.
}
\label{fig:benes-routing}
\end{figure}

\begin{lemma}[\cite{Benes64a}]
    \label{lem:benes-route-correct}
    Given as input $\ell \in \NN$ and a perfect matching $M$ between the $s = 2^\ell$ inputs and outputs of $B_\ell$,
    the procedure \textnormal{\textsc{BenesLink}}$(\ell,M)$ in \Cref{alg:route} computes an $M$-linkage in $B_\ell$ in time $O(s \log s)$.
\end{lemma}

\begin{proof}
    Consider the \emph{conflict graph} $D$ defined in Line~\ref{line:conflict}.
    Clearly, the edge set of $D$ is a union of two matchings, i.e., $D$ is $2$-edge-colorable (see Figure \ref{fig:benes-routing}).
    Hence, $D$ is bipartite and we can indeed compute a proper $2$-(vertex)-coloring $f\colon [s] \to \{\uparrow,\downarrow\}$ in Line~\ref{line:conflict-bipartite}.

    Since $f$ is a $2$-coloring of $D$, the pairs $\{\{v_i,v_{t(i)}^{f(i)}\} \mid i \in [s]\}$ form a matching.
    Indeed, for $i \in [s/2]$, the vertex $v_i^{\uparrowsm}$ (and $v_i^{\downarrowsm}$, respectively) can only be paired with $v_i$ and $v_{i+s/2}$.
    Because $\{i,i+s/2\}$ is an edge of $D$, we conclude that $f(i) \neq f(i + s/2)$ and hence, $v_i^{\uparrowsm}$ (and $v_i^{\downarrowsm}$, respectively) is only paired with one of the vertices $v_i$ and $v_{i+s/2}$.
    Similarly, the pairs $\{\{w_{t(j)}^{f(i)},w_j\} \mid i \in [s], j = \pi(i)\}$ form a matching.

    It follows that the sets $M^{\uparrowsm}$ and $M^{\downarrowsm}$ form matchings between the inputs and outputs of $B^\uparrowsm_{\ell-1}$ and $B^\downarrowsm_{\ell-1}$, respectively, and the paths computed in Line~\ref{line:build-path} form an uncongested $M$-linkage.
    
    It remains to analyze the running time $T(s)$ of \textsc{BenesLink}$(\ell,M)$ for $s=2^\ell$ in the word RAM model. 
    Clearly, all operations except for the recursive calls in Lines~\ref{line:B-up-linkage} and~\ref{line:B-down-linkage} cost $\Theta(s)$ time in total. 
    We obtain the recurrence $T(s) = 2T(s/2)+\Theta(s) = \Theta(s\log s)$.
\end{proof}

\begin{proof}[Proof of \Cref{thm:augmented-benes-linkage}]
    Let $\ell \in \NN$ and $s = 2^\ell$. Denote the inputs and outputs of $B_\ell$ and $\check{B}_\ell$ by $V$ and $W$.
    Given a perfect matching $\check{M} = \{e_1, \ldots, e_{s/2}\}$ on $V$, we construct an $\check{M}$-linkage in $\check{B}_\ell$.

    First, define from $\check{M}$ a perfect matching $M$ between $V$ and $W$ to be used in $B_\ell$:
    Let $W = \{w_1 ,\ldots, w_s\}$. For $i \in [s/2]$, write $e_i = ab$ and include edges $aw_{2i-1}$ and $bw_{2i}$ into $M$.
    \Cref{lem:benes-route-correct} finds an $M$-linkage $Q$ in the plain network $B_\ell$ in time $O(s \log s)$.
    We patch $Q$ to an $\check{M}$-linkage $\check{Q}$ in $\check{B}_\ell$: For $ab \in \check{M}$, writing $w$ and $w'$ for the unique partners of $a$ and $b$ in $M$,
    we add the concatenated path $\check{P}_{ab} = P_{aw} \circ (w,w') \circ P_{w'b}$.
    Because all paths $\check{P}_{ab}$ exist in $\check{B}_\ell$ and are vertex-disjoint, it follows that $\check{Q} = \{\check{P}_{ab} \mid ab \in \check{M} \}$ is an $\check{M}$-linkage in $\check{B}_\ell$. 
\end{proof}

\section{Linkages in Random Graphs} \label{app:random-graph}

We use $\mathcal{G}(k,p)$ for the Erd\H{o}s-R\'{e}nyi random graph model with edge probability $p$ on $k$ vertices, and $G(k,m)$ for the uniform distribution over all graphs with $k$ vertices and $m$ edges. 

An \emph{equipartition} of a set $M$ into $r$ parts is a partition $M_1,\dots,M_r$ such that $|M_i-M_j|\leq 1$ for all $i,j\in[r]$. 

\begin{theorem}[{\cite[Corollary 1.1]{BFSU96}}]
    \label{theo:random:pair:linked}
    There exists an absolute constant $\beta>0$ that the following holds. 
    Let $\varepsilon > 0$ be a constant and $d(k) \geq (1 + \varepsilon)\log(k)$. 
    With high probability, for random $H \sim G (k, m)$ with even $k$ and $m = k\cdot d(k)/2$ and any perfect matching $M$ on vertices $[k]$, with high probability a random equipartition of $M$ into $r = \lceil \beta \log(k) / \log(d) \rceil$ matchings $M_1, \dots, M_r$ satisfies that $H$ contains an $M_i$-linkage for all $i\in [r]$.
\end{theorem}

We first transfer the above result from the $G(k,m)$ model to the $\mathcal{G}(k,p)$ model. 

\begin{corollary} \label{cor:BFSU_gnp}
There exists an absolute constant $\beta>0$ that the following holds. 
Let $\varepsilon' > 0$ be a constant and $p \geq (1+\varepsilon')\log(k)/k$. 
With high probability, for random $H \sim \mathcal G (k, p)$ with even $k$ and any perfect matching $M$ on vertices $[k]$, with high probability a random equipartition of $M$ into $r = \lceil \beta \log(k) / \log(kp) \rceil$ matchings $M_1, \dots, M_r$ satisfies that $H$ contains an $M_i$-linkage for all $i\in [r]$.
\end{corollary}

\begin{proof}
Let $\mathcal{P}$ be the graph property specified in \Cref{theo:random:pair:linked} that $G(k,m)$ satisfies with high probability, and $\overline{\mathcal{P}}$ be the negation of $\mathcal{P}$.  
We show that, in the setting of this corollary, any graph drawn from the $\mathcal G(k,p)$ model satisfies $\mathcal{P}$ with high probability. 

Note that $\mathcal{P}$ is a monotone increasing property, meaning if $H$ satisfies $\mathcal{P}$ then $H+e$ satisfies $\mathcal{P}$ too. 
Therefore, $\overline{\mathcal{P}}$ is monotone decreasing. 
By coupling, it holds that 
\begin{equation}\label{equ:monotone-property}
\Pr_{H\sim G(k,m_1)}[H\in\overline{\mathcal{P}}]\leq\Pr_{H\sim G(k,m_2)}[H\in\overline{\mathcal{P}}]
\qquad\text{for $m_1\geq m_2$.}
\end{equation}

Take $\varepsilon=\varepsilon'/3$, and set $m^*:=\frac{1+\varepsilon}{1+2\varepsilon} p\cdot\binom{k}{2}$. 
Draw a graph $H\sim\mathcal G(k,p)$. By the law of total probability, we have
\begin{equation} 
\Pr_{H\sim\mathcal G(k,p)}[H\in\overline{\mathcal{P}}]\leq \underbrace{\Pr_{H\sim\mathcal G(k,p)}[|E(H)|<m^*]}_{\text{\circlednum{1}}}+\underbrace{\Pr_{H\sim\mathcal G(k,p)}[H\in\overline{\mathcal{P}}\land |E(H)|\geq m^*]}_{\text{\circlednum{2}}}. 
\end{equation}
We bound term \circlednum{1} by a standard Chernoff bound
\[
\text{\circlednum{1}}=\Pr\left[\mathrm{Bin}\left(\binom{k}{2},p\right)<\left(1-\frac{\varepsilon}{1+2\varepsilon}\right)\binom{k}{2}p\right]\leq\exp\left\{-\binom{k}{2}p\cdot\frac{\varepsilon^2}{2(1+2\varepsilon)^2}\right\}=O(k^{-k}). 
\]
We bound term \circlednum{2} by the following. 
\begin{align*}
\text{\circlednum{2}}&=\sum_{t=m^*}^{\binom{k}{2}}\Pr_{H\sim\mathcal G(k,p)}\left[H\in\overline{\mathcal{P}} \;\big|\; |E(H)|=t\right]\cdot\Pr_{H\sim\mathcal G(k,p)}[|E(H)|=t]\\
&=\sum_{t=m^*}^{\binom{k}{2}}\Pr_{H'\sim G(k,t)}[H'\in\overline{\mathcal{P}}]\cdot\Pr_{H\sim\mathcal G(k,p)}[|E(H)|=t]\\
&\leq\sum_{t=m^*}^{\binom{k}{2}}\Pr_{H'\sim G(k,m^*)}[H'\in\overline{\mathcal{P}}]\cdot\Pr_{H\sim\mathcal G(k,p)}[|E(H)|=t] \tag*{(by \eqref{equ:monotone-property})}\\
&=\Pr_{H'\sim G(k,m^*)}[H'\in\overline{\mathcal{P}}]\cdot\sum_{t=m^*}^{\binom{k}{2}}\Pr_{H\sim\mathcal G(k,p)}[|E(H)|=t]\leq \Pr_{H'\sim G(k,m^*)}[H'\in\overline{\mathcal{P}}]. 
\end{align*}
Because
\[
\frac{2m^*}{k}\geq \left(\frac{1+3\varepsilon}{1+2\varepsilon}\cdot\frac{k-1}{k}\right)\cdot(1+\varepsilon)\log k, 
\]
we invoke \Cref{theo:random:pair:linked} with constant $\varepsilon$ for large enough $k$, to see that \circlednum{2} is also negligible. 
\end{proof}

\begin{proof}[Proof of \Cref{thm:BFSU-gnp}]
Assume $k$ is even. We extend the matching $M$ to a perfect matching $M'$ by pairing the unmatched vertices. 
By \Cref{cor:BFSU_gnp}, a random equipartition of $M'$ into $M'_1,\dots,M'_r$ satisfies the desired property. 
We then drop the edges in $M'\setminus M$ from this partition to obtain $M_1,\dots,M_r$ with the desired property. 

Assume $k$ is odd. Given the matching $M$, we find an arbitrary unmatched vertex $w\in V(H)$. The induced subgraph $H-w$ is subject to the uniform distribution $\mathcal{G}(k-1,p)$. 
In the regime of \Cref{cor:BFSU_gnp}, if $p=p(k)\geq (1+\varepsilon')\log(k)/k$, then $p(k-1)\geq (1+\varepsilon'')\log(k)/k$ for some other constant $\varepsilon''>0$. 
Therefore, we can invoke \Cref{cor:BFSU_gnp} again, and the rest of the argument is the same as the even $k$ case. 
\end{proof}

\section{Counting Small Induced Subgraphs}
\label{sec:app-indsub}

We give a proof of \cref{thm:indsub-hardness-eth}, relying on \cite[Lemma 3.3 \& A.3]{CurticapeanN24}, stated below.

\begin{lemma}[{\cite[Lemma 3.3]{CurticapeanN24}}]
    \label{lem:phi=sub}
    Given a $k$-vertex graph invariant $\Phi$, for $k \geq 1$, we have
    \begin{equation}
        \numindsub{\Phi}{\star} = \sum_H \widehat{\Phi}(H) \cdot \numsub{H}{\star},
    \end{equation}
    where $H$ ranges over all unlabelled $k$-vertex graphs.
\end{lemma}

Let $G = (V,E,c)$ be a colored graph where $c\colon V(G) \to C$.
We define $G^{\circ} = (V,E)$ to be the uncolored version of $G$.
For $i\in C$, we write $V_{i}(G)$ for the vertices of color $i$, and for $i,j\in C$, we write $E_{ij}(G)$ for the edges in $G$ with one endpoint of color $i$ and another of color $j$.
For $X \subseteq C$ and $Y \subseteq \binom{C}{2}$, let $G_{\setminus X,Y}$ be the graph obtained from $G$ by deleting all vertices with colors from $X$ and all edges whose endpoints have a color pair from $Y$, i.e,
\[G_{\setminus X,Y}=\left(V\setminus\bigcup_{i\in X}V_{i},\ E\setminus\bigcup_{ij\in Y}E_{ij}\right).\]
The following observation is immediate.

\begin{observation}
    \label{obs:colsub-preprocess}
    For graphs $H$ and $G$ with canonically colored $H$, we have 
    \[\numsub{H}{G}=\numsub{H}{G_{\setminus\emptyset,\overline{E(H)}}}.\]
\end{observation}

Also, we write $G \cong H$ to denote that two (colored or uncolored) graphs $G,H$ are isomorphic.

\begin{lemma}[{\cite[Lemma A.3]{CurticapeanN24}}]
    \label{lem:sub-monotonicity}
    Let $k\in\mathbb{N}$ and let the following be given:
    \begin{itemize}
        \item Numbers $\alpha_{1},\ldots,\alpha_{s}\in\mathbb{Q}$ and pairwise non-isomorphic uncolored graphs $H_{1},\ldots,H_{s}$ with $|V(H_{i})|=k$ for all $i\in[s]$, which define the graph invariant
        \[
        f(\star) \coloneqq \sum_{i=1}^{s}\alpha_{i} \cdot \numsubstar{H_{i}},
        \]
        \item a canonically colored graph $H$ with $V(H) = [k]$ and $H^{\circ} \cong H_{b}$ for some $b\in[s]$, and 
        \item a colored graph $G$ with coloring $c:V(G)\to[k]$ satisfying $E_{ij}(G)=\emptyset$ for $ij\notin E(H)$. 
    \end{itemize}
    Then we have
    \[
    \alpha_{b}\cdot\numsub{H}{G}=\sum_{\substack{X\subseteq V(H)\\
    Y\subseteq E(H)
    }
    }(-1)^{|X|+|Y|}f(G_{\setminus X,Y}^{\circ}).
    \]
\end{lemma}

With these tools at our disposal, we are ready to prove \cref{thm:indsub-hardness-eth}.

\begin{proof}[Proof of \cref{thm:indsub-hardness-eth}]
    Let $\Phi$ be a $k$-vertex graph invariant and suppose $H$ is a graph with $\widehat{\Phi}(H) \neq 0$ and $E(H) \geq k \cdot \ell \geq N_0$.
    Without loss of generality assume $V(H) = [k]$.
    We give an algorithm for $\ccolsub{H}$ that uses an algorithm for $\sindsub{\Phi}$ as a subroutine.
    
    Let $G$ be the input graph for the problem $\ccolsub{H}$.
    We may assume that $G = G_{\setminus\emptyset,\overline{E(H)}}$ by \cref{obs:colsub-preprocess}.
    We wish to determine $\numsub{\can{H}}{G}$. 
    By Lemma~\ref{lem:phi=sub} we have
    \[\numindsub{\Phi}{\star} = \sum_F \widehat{\Phi}(F) \cdot \numsub{F}{\star} =: f(\star),\]
    where $F$ ranges over all $k$-vertex graphs.
    Invoking Lemma~\ref{lem:sub-monotonicity}, we obtain that
    \begin{equation}
        \label{eq:using-submon-in-reduction}
        \widehat{\Phi}(H) \cdot \numsub{\can{H}}{G} = \sum_{\substack{X\subseteq V(H)\\Y\subseteq E(H)}}(-1)^{|X|+|Y|}f(G_{\setminus X,Y}^{\circ}).
    \end{equation}
    So we can compute $\numsub{\can{H}}{G}$ by evaluating the right-hand side of \eqref{eq:using-submon-in-reduction} and dividing by $\widehat{\Phi}(H) \neq 0$.
    Note that all relevant values $f(G_{\setminus X,Y}^{\circ})$ can be obtained by the oracle calls $\numindsub{\Phi}{G^\circ_{\setminus X,Y}}$ without parameter increase in overall time $2^{|V(H)|+|E(H)|} \cdot n^{O(1)}$.
    The value $\widehat{\Phi}(H)$ can be computed by brute-force by evaluating $\Phi$ on $2^{O(k^2)}$ many $k$-vertex graphs.
    Hence, an $O(n^{\beta \cdot \ell})$ algorithm for $\sindsub{\Phi}$ gives an $O(n^{c \cdot \beta \cdot \ell})$ for $\ccolsub{H}$ for some suitable fixed constant $c$.
    Now, the theorem follows from Theorem \ref{thm:dense-hard}. 
\end{proof}

\end{document}